\documentclass[twoside,11pt]{article}

\usepackage{jmlr2e}
\usepackage{amsmath}
\usepackage{booktabs}
\usepackage{bbm}
\usepackage{algorithm}
\usepackage{algorithmic}
\usepackage{tcolorbox}
\usepackage{mdframed}
\usepackage{array}
\usepackage{geometry}
\usepackage{multirow}
\usepackage{todonotes}
\usepackage{mathrsfs}

\DeclareMathOperator{\Proj}{Proj}
\DeclareMathOperator{\Span}{span}
\DeclareMathOperator{\hull}{Hull}
\DeclareMathOperator{\proc}{Proc}

\ShortHeadings{Formal Privacy Guarantees with Invariant Statistics}{Cho and Awan}
\firstpageno{1}

\begin{document}

\title{Formal Privacy Guarantees with Invariant Statistics}

\author{\name Young Hyun Cho \email cho472@purdue.edu \\
       \addr Department of Statistics\\
       Purdue University\\
       150 N University St, West Lafayette, IN 47907
       \AND
       \name Jordan Awan \email jawan@purdue.edu \\
       \addr Department of Statistics\\
       Purdue University\\
       150 N University St, West Lafayette, IN 47907}

\editor{}

\maketitle

\begin{abstract}
Motivated by the 2020 US Census products, this paper extends differential privacy (DP) to address the joint release of DP outputs and nonprivate statistics, referred to as invariant. Our framework, Semi-DP, redefines adjacency by focusing on datasets that conform to the given invariant, ensuring indistinguishability between adjacent datasets within invariant-conforming datasets. We further develop customized mechanisms that satisfy Semi-DP, including the Gaussian mechanism and the optimal \(K\)-norm mechanism for rank-deficient sensitivity spaces. Our framework is applied to contingency table analysis which is relevant to the 2020 US Census, illustrating how Semi-DP enables the release of private outputs given the one-way margins as the invariant. Additionally, we provide a privacy analysis of the 2020 US Decennial Census using the Semi-DP framework, revealing that the effective privacy guarantees are weaker than advertised. 
\end{abstract}


\section{Introduction}\label{sec: introduction}

In the era of big data and pervasive digital tracking, safeguarding individual privacy has become increasingly critical. Differential Privacy (DP) \citep{dwork2006calibrating} has emerged as the gold standard for ensuring strong privacy protections while enabling meaningful data analysis. A key advantage of DP is its ability to provide individual-level protection, offering both data curators and users a straightforward understanding on the privacy protection. DP achieves this by ensuring that the outputs from two adjacent databases---databases that differ by a single entry---have similar distributions, with this similarity quantified by a privacy parameter. To maintain DP, data curators introduce calibrated noise based on the maximum difference in outputs that adjacent databases can produce, effectively masking the contribution of any individual entry. This privacy framework has been widely adopted by major tech companies such as Google \citep{erlingsson2014rappor}, Apple \citep{team2017learning}, and Microsoft \citep{ding2017collecting}, as well as in the public sector, where protecting sensitive information is crucial, exemplified by the practices of the US Census Bureau \citep{abowd20222020}.

While DP has proven effective in many contexts, real-world applications often demand the publication of true statistics on confidential data alongside DP outputs, particularly when such statistics are necessary or when the perceived public benefits are substantial. For instance, in the 2020 US Decennial Census, the US Census Bureau provided noisy demographic counts that satisfy DP but also disclosed certain counts without noise, such as the total resident populations of states for electoral representation purposes \citep{abowd20222020}. These accurate counts, referred to as invariants, reveal true information about the confidential data, potentially compromising privacy.

The joint release of DP outputs and invariants differ from standard DP scenarios, and privacy guarantees in these settings are still underexplored. Motivated by the 2020 Census, in \cite{gao2022subspace}, Subspace Differential Privacy (Subspace DP) was introduced and further studied in \cite{dharangutte2023integer}. It is applicable when the invariant is a linear transformation of a query, and considers DP on subspaces unaffected by the linear transformation. While Subspace DP provides some insight into what the Census mechanism protects, it is not the ideal framework for quantifying the privacy protection. Like traditional DP, it considers databases differing by a single entry, but these adjacent databases may yield different invariants, allowing an adversary to identify confidential data. This is particularly problematic when the invariants are count statistics, as seen in the 2020 US Census as databases differing by a single entry have different counts.

In response, we build upon the framework of semi-differential privacy (Semi-DP), initially introduced in \cite{awan2023canonical}, to address the challenges posed by invariant. A key observation is that the release of invariant constrains the set of possible databases that can be protected. To account for this, we explicitly consider invariant-conforming databases and redefine the adjacency relation to ensure indistinguishability within this reduced set. Our approach guarantees that the confidential data remains indistinguishable from any adjacent database within the set of invariant-conforming databases, even when both the invariants and the DP outputs are released jointly. Additionally, we introduce customized mechanism designs that data curators can use to publish specific invariants and DP outputs in accordance with our framework.

\subsection{Related Literature}
 Several studies have explored differential privacy in the context of public information, notably subspace differential privacy as introduced in \cite{gao2022subspace}, as well as the work of \cite{gong2020congenial} and \cite{seeman2022formal}. Both \cite{gao2022subspace} and \cite{gong2020congenial} assume that the invariants are linear transformations of the query which is privatized with a DP mechanism. This aligns well with settings like the 2020 Decennial Census where certain aggregate statistics are published without noise. However, this assumption does not apply in scenarios where no linear relationship exists between the DP query and the invariant. For example, consider a biomedical research setting where true demographic counts are released without noise while regression coefficients are privatized with DP. In such cases, the frameworks proposed in the prior works are not be applicable.

\cite{seeman2022formal} introduced the \((\epsilon,\{\Theta_{z}\}_{z \in \mathcal{Z}})\)-Pufferfish framework, which extends Pufferfish privacy \citep{kifer2014pufferfish} by incorporating a collection of mechanisms indexed by the realizations of a random variable \( Z \). This variable represents public information related to the confidential data that is not protected by the DP mechanism. While this approach offers nuanced privacy protection, it introduces complexity and practical challenges, making it difficult to implement and interpret. Furthermore, it does not operate strictly within the traditional DP framework.

Our work diverges from these approaches by focusing on invariant-conforming databases, where traditional DP adjacency relations may not hold. Building upon the Semi-DP example in \cite{awan2023canonical}, we draw inspiration from \cite{awan2021structure} and \cite{kim2022differentially} in designing mechanisms that inject noise based on sensitivity space. \cite{awan2021structure} introduced sensitivity space as a key tool for mechanism design, which we adapt to the invariant-conforming setting. Additionally, \cite{kim2022differentially} provided insights on reducing noise when the sensitivity space is rank-deficient, an idea we generalize to all common additive mechanisms.

Finally, to avoid confusion in terminology, it is important to note that \cite{lowy2023optimal} also adopted the term Semi-DP. However, their work fundamentally differs from ours. In \cite{lowy2023optimal}, Semi-DP refers to private model training using public data that does not require protection. In contrast, our notion of Semi-DP originated from \cite{awan2023canonical}, where it was developed as a tool to bound the power of differentially private uniformly most powerful(UMP) tests by allowing some summary statistics to be released without privacy protection. As suggested in \cite{awan2023canonical}, we demonstrate that Semi-DP is of independent interest, providing a novel framework for analyzing privacy protection in scenarios, which involve the joint release of invariant and DP mechanisms. 

\subsection{Our Contributions}

This paper addresses a critical gap in privacy research by introducing Semi-DP, a novel extension of DP tailored for scenarios where true statistics, or invariants, are released alongside DP outputs. Semi-DP redefines adjacency to operate within invariant-conforming datasets, thereby safeguarding sensitive information even when true statistics are disclosed. Our key contributions are threefold:

\begin{itemize}
    \item {\bf Extension of Semi-DP Framework:} Building upon the Semi-DP framework originally introduced in \cite{awan2023canonical}, we significantly extend its application to address the joint release of arbitrary DP outputs and invariants. Our extension provides a more comprehensive privacy analysis for scenarios where traditional DP models fall short, offering greater clarity in understanding privacy risks in complex data scenarios.

    \item {\bf Optimal Mechanism Design:} We develop customized mechanisms that satisfy the Semi-DP framework, including the optimal \(K\)-norm mechanism, which is particularly effective for rank-deficient sensitivity spaces often encountered due to invariant. Our approach retains the core structure of DP, allowing existing DP methods like Laplace, Gaussian, and \(K\)-norm mechanisms to be adapted with minimal modifications. By leveraging the sensitivity space, our mechanisms minimize noise addition while preserving privacy. Moreover, we propose a private uniformly most powerful unbiased (UMPU) test for the odds ratio in a $2 \times 2$ contingency table that satisfies Semi-DP.

    \item {\bf Application and Analysis on Census Data:} We apply our Semi-DP framework to contingency table analysis, which is particularly relevant to the 2020 US Census, demonstrating that private outputs can be released while preserving true marginal counts. Furthermore, we conduct a detailed privacy analysis of the 2020 US Decennial Census to illustrate the practical applicability of our framework in real-world scenarios.

\end{itemize}

Through these contributions, our work not only generalizes existing DP methods but also provides a versatile and practical solution for privacy analysis in scenarios with complex relationships between invariants and privatized outputs.

\subsection{Paper Organization}
In Section \ref{sec: preliminary}, we review preliminaries on DP, including adjacency relations, sensitivity space, and key mechanisms like Gaussian mechanism and $K$-norm mechanism. In Section \ref{sec: semi DP}, we introduce Semi-Differential Privacy (Semi-DP), an extension of differential privacy designed to address the joint release of DP outputs and invariant. In Section \ref{sec: mechanism design}, we develop tailored mechanisms for Semi-DP, such as Gaussian mechanism and $K$-norm mechanisms, focusing on optimizing noise through sensitivity space analysis. In Section \ref{sec: contingency table}, we apply Semi-DP mechanisms to output a noisy contingency table given the one-way margins as invariant. In Section \ref{sec: census}, we provide a privacy analysis on 2020 US Census based on our Semi-DP framework. Finally, in Section 7, we conclude this paper with some discussion. We relegate all the proofs to Appendix \ref{sec: appendix proof}.

\section{Preliminaries on Differential Privacy}\label{sec: preliminary}

In this section, we review the fundamental concepts of DP, focusing on the notions of adjacency and the role of hypothesis testing in quantifying privacy protection. We also discuss key properties such as group privacy, composition, and post-processing invariance. Additionally, we introduce the concept of sensitivity space and its significance in guiding mechanism design, particularly in the context of adding noise to queries to ensure DP.

\subsection{Adjacency in Differential Privacy}

DP relies on an adjacency relation on databases, which quantifies whether two databases are ``differing in one entry.'' Adjacency is crucial because DP requires that the outputs generated from two adjacent datasets to be similar. In many DP applications, adjacency is defined in terms of a metric on a given dataspace. A common example is when the input is tabular data, and the Hamming distance is used. If the size of the tabular data is known to be \(n\), the dataspace is \(\mathcal{D} = \mathcal{X}^{n}\), where each individual contribution is denoted by \(\mathcal{X}\). The adjacency relation is then defined using the Hamming distance as \(A(X,X’) = \mathbbm{1}_{\{H(X,X’) = 1\}}\).

Unlike traditional DP, our approach focuses on a subset of the entire dataspace, specifically the invariant-conforming databases that share the same invariant values as the confidential data. Since we define adjacency only among these databases, it is helpful to outline the distance metrics we will consider for this purpose.

\begin{definition}
    Let \(\mathcal{D}\) be a dataspace and consider a metric \(d: \mathcal{D} \times \mathcal{D} \rightarrow \mathbb{Z}^{\geq 0}\). Define an \emph{adjacency function} \(A(X,X'): \mathcal{D}\times \mathcal{D}\rightarrow \{0,1\}\) by \(A(X,X') = \mathbbm{1}_{\{d(X,X')=1\}}\). We say that a metric \(d\) is an \emph{adjacency metric} if \(d(X, X') = a\) is equivalent to the existence of a sequence of databases \(X_{0} = X, X_{1}, \dots, X_{a} = X'\) such that \(A(X_{i-1}, X_{i}) = 1\) for all \(i \in [a]\), and $a$ is the smallest value for which this is possible.
\end{definition}

\begin{example}
Most widely accepted metrics in DP literature indeed are adjacent metrics: For Hamming distance $H$ defined on dataspace $\mathcal{X}^{n}$, \(H(X, X')\) represents the number of different entries between two datasets having the same total number of entries. Then for $X,X'$ such that $H(X,X')=a$, there exists a sequence of datasets \(X_0 = X, X_1, \dots, X_a = X'\), where for each pair \(X_{i-1}\) and \(X_i\), $X_{i}$ can be obtained by switching an entry in $X_{i-1}$ to some other entry, and so satisfying \(A(X_{i-1}, X_{i}) = 1\). This sequence confirms that the Hamming distance is an adjacency metric.

On the other hand, when the total number of rows is not known, it is usual to consider the dataspace to be \(\mathcal{D} = \bigcup_{i = 0}^{\infty}\mathcal{X}^{i}\). The adjacency relation in this case is typically defined as \(A(X,X’) = \mathbbm{1}_{\{X \Delta X’ = 1\}}\), where \(X \Delta X’\) represents the symmetric difference between two datasets. The criterion here is whether one dataset can be obtained from the other by adding or removing a single entry. In this case, \(X \Delta X' = a\) indicates that \(X'\) can be obtained by adding/deleting a total of \(a\) rows from \(X\). We see that the symmetric difference distance is also an adjacency metric.
\end{example}

\begin{remark}
    In DP, the adjacency relation is typically defined between databases that differ by a distance of 1, capturing the idea that they vary by only a single entry. The focus on such adjacent databases enables the interpretation of DP as providing individual-level protection. This individual-level protection enhances the interpretability of DP, allowing both the data curator deploying the DP mechanism and the individuals whose data is included to clearly understand what aspects of their data are being protected. In this paper, we will design an adjacency function for Semi-DP with the same goal of individual-level interpretability as well.
\end{remark}

\subsection{Differential Privacy via Hypothesis Testing}

Once adjacency is understood, we can explore the privacy protection provided by DP. Over time, various DP variants have introduced different methods for quantifying similarity between distributions, such as likelihood ratios \citep{dwork2006differential}, R\'enyi divergence \citep{bun2016concentrated,mironov2017renyi}, and hypothesis testing \citep{dong2022gaussian}. Hypothesis testing, in particular, is a powerful tool for understanding DP, where the objective of the adversary is to distinguish between two adjacent input datasets, $X$ and $X'$. The quantification of privacy protection relies on bounding the type I and type II errors of the most powerful test for this hypothesis. In \cite{dong2022gaussian}, the tradeoff function $T(P,Q): [0,1] \rightarrow [0,1]$ between distributions $P$ and $Q$ is introduced, defined as $T(P,Q)(\alpha) = \inf \{1-\mathbb{E}_{Q}\phi : \mathbb{E}_{P}(\phi) \leq 1 - \alpha\}$, where the infimum is taken over all measurable tests $\phi$.\footnote{The tradeoff function in \cite{dong2022gaussian} was initially defined as the smallest type II error given a type I error of $\alpha$. We follow \cite{awan2022log} and \cite{awan2023optimizing}, who use type I error $1-\alpha$ as this simplifies the formula for group privacy.} It provides the optimal type II error for testing $H_{0} = P$ versus $H_{1} = Q$ at type I error $1-\alpha$, thereby capturing the difficulty of distinguishing between $P$ and $Q$.

\begin{definition}\label{def: conventional f dp}
    ($f$-DP: \cite{dong2022gaussian}) Let \(\mathcal{D}\) be a dataspace, \(A: \mathcal{D} \times \mathcal{D} \rightarrow \{0,1\}\) be an adjacency function, and \(f\) be a tradeoff function. A mechanism \(M(X)\) is \((\mathcal{D},A,f)\)-DP if 
    \begin{equation*}
        T(M(X),M(X')) \geq f,
    \end{equation*}
    for all \(X, X' \in \mathcal{D}\) such that \(A(X,X') =1\).
\end{definition}

In Definition \ref{def: conventional f dp}, we explicitly include the underlying dataspace and adjacency function, unlike the typical expression that focuses solely on the tradeoff function $f$ as the privacy parameter. This is because our framework will account for a invariants and specifically considers adjacent databases among invariant-conforming databases. Therefore, it is essential to explicitly define the dataspace and adjacency function to accurately capture the privacy guarantee in our framework. 

Intuitively speaking, a mechanism \(M(X)\) satisfies \(f\)-DP, where \(f = T(P,Q)\), if testing \(H_{0}: X\) versus \(H_{1}: X'\) is at least as difficult as testing \(H_{0}: P\) versus \(H_{1}: Q\). The tradeoff function \(f\) serves as the privacy parameter, guiding the level of privacy protection. Without loss of generality we can assume $f$ is \emph{symmetric}, meaning that if $f = T(P,Q)$, then $f = T(Q,P)$. This is due to the fact that adjacency of databases is a symmetric relation \citep{dong2022gaussian}. Moreover, a function $f: [0,1] \rightarrow [0,1]$ is a tradeoff function if and only if $f$ is convex, continuous, non-decreasing, and $f(x) \leq x$ for all $x \in [0,1]$. We say that a tradeoff function $f$ is \emph{nontrivial} if $f(\alpha) < \alpha$ for some $\alpha \in (0,1)$. 

Two important special cases of \(f\)-DP are \((\mathcal{D},A,f_{\epsilon,\delta})\) and \((\mathcal{D},A,G_{\mu})\), which correspond to $(\epsilon,\delta)$-DP and $\mu$-Gaussian DP(GDP), respectively. For \(\epsilon \geq 0\) and \(\delta \in [0,1]\),  we define \(f_{\epsilon,\delta}(\alpha) = \max\{0, 1 - \delta - e^{\epsilon} + e^{\epsilon}\alpha, e^{-\epsilon}(\alpha - \delta)\}\), and for \(\mu \geq 0\), we define \(G_{\mu} = T(N(0,1),N(\mu,1))\). Moreover, a lossless conversion from \((\mathcal{D}, A, G_{\mu})\)-DP to a family of \((\mathcal{D},A,f_{\epsilon,\delta})\)-DP guarantees was developed in \cite{balle2018improving}: a mechanism is \((\mathcal{D}, A, G_{\mu})\)-DP if and only if it is \((\mathcal{D},A,f_{\epsilon,\delta(\epsilon)})\)-DP for all $\epsilon \geq 0$, where $\delta(\epsilon) = \Phi \left(-\frac{\epsilon}{\mu}+\frac{\mu}{2}\right) - e^{\epsilon} \Phi \left(-\frac{\epsilon}{\mu}-\frac{\mu}{2}\right)$. 

DP has the following three properties:\\

\noindent
{\bf Group Privacy:} Group privacy addresses the protection of privacy when a group of size \(k\) in a database is replaced by another group of the same size. When using an adjacency metric \(d\), group privacy concerns data pairs where \(d(X, X') \leq k\). Representing this through the adjacency function, we define \(A_{k}(X, X’) = \mathbbm{1}_{\{d(X, X’) \leq k\}}\) to capture the indistinguishability between \(X\) and \(X'\) under these conditions. Additionally, the level of privacy protection diminishes as the group size \(k\) increases. In particular, for a mechanism \(M\) that satisfies \((\mathcal{D}, A_{1}, f)\)-DP, it satisfies \((\mathcal{D}, A_{k}, f^{\circ k})\)-DP, where \(f^{\circ k}\) represents the functional composition of \(f\) applied \(k\) times. In the case of GDP, there is a particularly convenient formula: a mechanism M that satisfies $(\mathcal{D},A_{1},G_\mu)$-DP satisfies $(\mathcal{D},A_{k},G_{k\mu})$-DP.\\

\noindent
{\bf Composition:} The \(f\)-DP framework can quantify the cumulative privacy cost for the sequential release of multiple DP outputs. To model this process, let \(M_1: \mathcal{D} \to Y_1\) be the first mechanism applied to a dataset \(X\), and define \(Y_1 = M_1(X)\). The output \(Y_1\) then serves as an input to a second mechanism \(M_2: \mathcal{D} \times Y_1 \to Y_2\), which generates \(Y_2 = M_2(X, y_1)\). This recursive process continues for \(m\) mechanisms, such that the \(i\)-th mechanism \(M_i: \mathcal{D} \times Y_1 \times \cdots \times Y_{i-1} \to Y_i\) produces \(Y_i = M_i(X, Y_1, \dots, Y_{i-1})\). 
Thus, the \(m\)-fold composed mechanism \(M: \mathcal{D} \to Y_1 \times \cdots \times Y_m\) is defined as \(M(X) = (Y_1, Y_2, \dots, Y_m)\), where each \(M_i\) satisfies \((\mathcal{D}, A, f_i)\)-DP. The entire composition $M(X)$ satisfies \((\mathcal{D}, A, f_1 \otimes \cdots \otimes f_m)\)-DP, where \(\otimes\) denotes the \emph{tensor product} of the tradeoff functions. Specifically, if \(f = T(P_1, Q_1)\) and \(g = T(P_2, Q_2)\), then \(f \otimes g = T(P_1 \times P_2, Q_1 \times Q_2)\), quantifying the overall privacy guarantee for the composed mechanism. \\

\noindent
{\bf Invariance to Post-Processing:} Additional data-independent transformations on the DP output does not weaken the DP guarantee. Formally, if $M: \mathcal{D} \rightarrow \mathcal{Y}$ is $(\mathcal{D},A,f)$-DP and $\proc: \mathcal{Y} \rightarrow \mathcal{Z}$ a (possibly randomized) function, $\proc(M(X)): \mathcal{D} \rightarrow \mathcal{Z}$ is also $(\mathcal{D},A,f)$-DP. 

\subsection{Mechanism Design}
The most widely used approach to achieve DP is to add independent noise to a query. To make it difficult to distinguish two adjacent input databases, the additive noise should be large enough to mask a difference in the query due to adjacent databases. To this end, the key tool for calibrating the noise is the sensitivity space and sensitivity.

\begin{definition}\label{def: sensitivity space} (Sensitivity Space: \cite{awan2021structure}) Let $\phi: \mathcal{D} \rightarrow \mathbb{R}^{d}$ be any function. The \emph{sensitivity space} of $\phi$ with respect to adjacency function $A$ is 
\begin{equation}\label{equ: sensitivity space}
    S_{\phi}^A = \left\{\phi(X)-\phi(X'): A(X,X') =1 \text{ where } X,X' \in \mathcal{D} \right\}.
\end{equation}
When it is clear from the context, we may omit the super- and subscripts of $S$.
\end{definition}

The sensitivity space in Definition \ref{def: sensitivity space} refers to the set of all possible differences in query outputs resulting from changing a dataset to an adjacent counterpart. It helps quantify the maximum deviation that adjacent datasets can have on the query, thereby guiding the amount of noise needed to ensure privacy.

\begin{definition}\label{def: k-norm sensitivity} ($K$-norm Sensitivity: \cite{awan2021structure}). For a norm $\Vert \cdot \Vert_{K}$, whose unit ball is $K = \{x:\Vert x \Vert_{K} \leq 1\}$, the $K$-norm sensitivity of $\phi$ with respect to adjacency function $A$ is $\Delta_{K}(\phi;A) = \sup_{u \in S_{\phi}^A}\Vert u \Vert_{K}.$ For $p \in [1,\infty]$, the $l_{p}$-sensitivity of $\phi$ is  $\Delta_{p}(\phi;A) = \sup_{u \in S_{\phi}^A}\Vert u \Vert_{p}$. When it is clear from the context, we may omit the specification of $A$ and simply write $\Delta_K(\phi)$ or $\Delta_p(\phi)$.
\end{definition}

The sensitivity of $\phi$, geometrically, corresponds to the largest radius of $S_{\phi}^A$, measured by the norm of interest, indicating the maximum change in $\phi$ due to adjacent databases \citep{awan2021structure}. Gaussian additive noise proportional to $\ell_{2}$-sensitivity leads to $\mu$-GDP, while the Laplace additive noise calibrated by $\ell_{1}$-sensitivity satisfies $\epsilon$-DP. Moreover, a multivariate extension of the Laplace mechanism are the $K$-norm mechanisms, which substitutes the $\ell_{1}$-norm with another norm.

\begin{proposition}\label{prop: standard gaussian mechanism} (Gaussian Mechanism:\cite{dong2022gaussian}). Define the Gaussian mechanism that operates on a query $\phi$ as $M(X) = \phi(X) + \xi$, where $\xi \sim N(0, \sigma I_{d})$ for some $\sigma \geq \mu^{-1}\Delta_{2}(\phi;A)$. Then $M$ is $(\mathcal{D},A,G_{\mu})$-DP.
\end{proposition}

\begin{proposition}\label{prop: K-norm mech}($K$-Norm Mechanism: \cite{hardt2010geometry}). 
Let $\Vert \cdot \Vert_{K}$ be any norm on $\mathbb{R}^{d}$ and let $K = \{x : \Vert x \Vert_{K} \leq 1\}$ be its unit ball. Let $\Delta_{K}(\phi;A) = \sum_{u \in S_{\phi}^A} \Vert u \Vert_{K}$. Let $V$ be a random variable in $\mathbb{R}^{d}$, with density $f_{V}(v) = \frac{\exp\left(\frac{\epsilon}{\Delta}\Vert v \Vert_{K}\right)}{\Gamma(d+1)\lambda\left(\frac{\Delta}{\epsilon}\right)}$, where $\Delta_{K}(\phi;A) \leq \Delta < \infty$. Then releasing $\phi(X) +V$ satisfies $(\mathcal{D},A,f_{\epsilon,0})$-DP.
\end{proposition}

\cite{awan2021structure} proposed a criteria to compare $K$-norm mechanisms:

\begin{definition}\label{def: optimal K-norm mech}(Containment Order and Optimal $K$-norm mechanism \citep{awan2021structure}). Let $V$ and $W$ be two random variables on $\mathbb{R}^{d}$ with densities $f_{V}(v) = c \exp\left(-\frac{\epsilon}{\Delta_{K}}\Vert v \Vert_{K}\right)$ and $f_{W}(w)= c \exp\left(-\frac{\epsilon}{\Delta_{H}}\Vert w \Vert_{H}\right)$. We say that $V$ is preferred over $W$ in the containment order if $\Delta_{K}\cdot K \subset \Delta_{H}\cdot H$. 

Moreover, We define the optimal $K$-norm mechanism as one that adds noise from a random variable \(V_{K}\) with density \(f_{V_{K}}(v) = c \exp\left(-\frac{\epsilon}{\Delta_{K}}\Vert v \Vert_{K}\right)\), where the norm ball \(\Delta_{K} \cdot K\) is the smallest in containment order among all norm balls that contain the sensitivity space \(S\). \end{definition}

In \cite{awan2021structure}, they showed that if $V$ is preferred over $W$ in the containment order, then $V$ has smaller entropy, has smaller conditional variance in any directions, and is stochastically tighter about its center. They also proved that the optimal $K$-norm mechanism is based on the convex hull of the sensitivity space, provided that the sensitivity space is full-rank, i.e., dim(span($S_{\phi}^A))=d$. Notably, $K_{\phi}^A := \hull(S_{\phi}^A)$ is shown to be a norm ball in $\mathbb{R}^{d}$ so that the norm $\Vert \cdot \Vert_{K_{\phi}^A}$ is well-defined the corresponding $K$-norm mechanism is optimal.

\begin{proposition}(Optimal $K$-norm mechanism for full-rank sensitivity space \citep{awan2021structure}).
    Let $\phi: \mathcal{D} \rightarrow \mathbb{R}^{d}$ such that $S_{\phi}^A$ is bounded and $\Span(S_{\phi}^A) = \mathbb{R}^{d}$. Then the $K$-norm mechanism using norm $\lVert \cdot \rVert_{K_{\phi}^A}$ and sensitivity $\Delta_K(\phi;A)=1$ is the optimal $K$-norm mechanism.    
\end{proposition}

However, given the potential violation of the full-rank condition due to invariants, we extend their findings in Section \ref{sec: optimal K norm mech} to provide the optimal $K$-norm mechanism applicable to rank-deficient sensitivity spaces, where dim(span($S_{\phi}))<d$.

\section{Semi Differential Privacy}\label{sec: semi DP}

This section introduces the Semi Differential Privacy (Semi-DP) framework, which extends traditional DP concepts to scenarios where invariants are jointly published with DP outputs. We start by demonstrating how an invariant can compromise standard DP protections, and explain how altering the set of adjacent databases results in adjusted privacy parameters. We contrast our approach with Subspace DP, highlighting its limitations in practical privacy contexts. Our framework focuses on preserving individual-level protection within invariant-conforming databases, ensuring a meaningful and interpretable privacy guarantee aligned with the core philosophy of DP.

\subsection{Challenges of Preserving DP with Invariants}\label{sec: problem setting}

To illustrate how the joint publication of an invariant can undermine privacy protection by the DP framework, we consider a simple example. 

\begin{example}\label{example: motivation}
    Suppose $\mathcal{D} = \{(x_{1},x_{2},x_{3}):x_{i} \in \{0,1\}$, \text{for all} $i=1,2,3\}$ and the true confidential data is $X^* = (0,1,1)$. While a data curator wants to publish $M(X)$ which is a differentially private version of a query $\phi(X)$, assume that the curator also wants to publish $T(X) = \#\{i \vert x_{i} = 1\}$ as well. In this case, the invariant is $t = T(X^*)=2$. 

    Without the invariant, DP protects any two databases differing in one entry. However, given the invariant $t=2$, an adversary would rule out all the databases that do not give the same invariant value and determine that there are only three databases which satisfy $T(X) = 2$: $(1,1,0)$, $(1,0,1)$, and $(0,1,1)$. Consequently, a privacy guarantee is only possible for these invariant-conforming databases. However, since the Hamming distance between any pair of the remaining databases is $2$, the privacy parameter should be modified as we have seen for group privacy. In other words, the presence of the invariant alters both the pairs protected by DP and the extent of that protection.
\end{example}

Example \ref{example: motivation} presents two key takeaways. The first is that a DP-style guarantee is only possible for invariant-conforming databases. From the adversary’s standpoint, if true information about the confidential data is provided by an invariant \(t\), the adversary would naturally aim to exploit \(t\) to its fullest extent. As a result, the adversary will immediately restrict their focus to invariant-conforming databases that align with \(t\), considering them the only possible true input databases. Therefore, regardless of the DP mechanism used for the query, this narrowing of candidates due to the provided invariant is inevitable.

Secondly, within invariant-conforming databases, the distance of $1$ can no longer serve as the standard for adjacent relations. In the case of Example \ref{example: motivation}, the Hamming distance among the remaining three databases is 2 for every pair. Therefore, it is necessary to redefine adjacency among invariant-conforming databases in order to give a meaningful DP guarantee.

\subsection{Limitations of Subspace DP}
In this section, we investigate subspace DP introduced in \cite{gao2022subspace} to highlight why focusing on invariant-conforming databases is crucial and necessary for privacy protection. The core issue of subspace DP is that it considers all adjacent databases---even those do not agree with the invariant. Note that the invariant in their consideration is a linear transformation of $\phi(X)$, or $C\phi(X)$ for some matrix $C$. Subspace DP can be expressed in $f$-
DP terminology, rather than $(\epsilon,\delta)$-DP in the original paper as follows: 
\begin{definition}\label{def: subspace dp}(Subspace Differential Privacy \citep{gao2022subspace}).
    Let $d$ be an adjacency metric and $f$ be a tradeoff function. Consider a query $\phi:\mathcal{D} \rightarrow \mathbb{R}^{d}$ and let $T(X) = C\phi(X)$ be a linear invariant for some $C$. A mechanism $M$ whose codomain is $\mathbb R^d$ satisfies $f$-subspace differential privacy if,  for any $X,X' \in \mathcal{D}$ such that $d(X,X') \leq 1$, 
    \begin{equation*}
T\left(\Proj_{\mathcal{C}}^{\perp}M(X),\Proj_{\mathcal{C}}^{\perp}M(X')\right) \geq f,
\end{equation*}
where $\mathcal{C} = \Span(S_{\phi})$ is the span of the sensitivity space and $\Proj_{\mathcal{C}}^{\perp}$ is the projection operator that projects onto the orthogonal complement of $\mathcal{C}$.
\end{definition}

Therefore, subspace DP considers noisy versions of $\phi(X)$ and only measures the closeness of outputs within the space where $\phi(X)$ is not already determined by the invariant $T(X)=C\phi(X)$. At a glance, this seems like a reasonable choice. If there are parts of $\phi(X)$ that $T(X)$ discloses without noise, there is no way to protect these overlapping parts regardless of any perturbation applied to $\phi(X)$. Therefore, the effort is directed towards protecting only the remaining parts that are not exposed. Additionally, by not adding noise to the parts that cannot be protected anyway, the total amount of noise introduced into the overall mechanism can be reduced, improving utility.

However, subspace DP cannot serve as a meaningful measure of privacy protection as it overlooks the fact that the adversary can also leverage the invariant \(T(X^*)=t\), where $X^*$ is the confidential data. The similarity in projected outputs does not provide any protection for \(X^*\) and \(X'\) when \(T(X^*) \neq T(X')\), as discussed above. 

Furthermore, a typical example of a linear transformation is a count statistic and databases differing by one entry cannot typically produce the same count statistics. 
Consequently, subspace DP fails to achieve indistinguishability between $X$ and $X'$, and so the $f$ parameter cannot be properly interpreted. This reinforces the argument to consider invariant-conforming databases instead of the original dataspace $\mathcal{D}$.

\subsection{Semi Differential Privacy}

Based on the previous discussion, we define our dataspace \(\mathcal{D}_{t}:= \{X \in \mathcal{D}: T(X) = t\}\) as the set of invariant-conforming databases that share the same invariant value as the confidential data, and we characterize a DP-style guarantee on this space. The term ``semi-private" was introduced in \cite{awan2023canonical}, where it is used to establish upper bounds on the power of differentially private testing by weakening the DP constraint. Specifically, in \cite{awan2023canonical}, semi-private testing refers to a test function where certain summary statistics of the confidential data are not protected, but DP is maintained for databases that share the same summary statistics as the confidential data. While we build on the notion of ``semi-private" from \cite{awan2023canonical}, we extend it to a broader setting by proposing a general DP-style protection for invariant-conforming databases in cases where both invariants and DP outputs are jointly released. 

First, suppose there exists a mechanism \(M\) that satisfies \((\mathcal{D}, A, f)\)-DP. We can then characterize the privacy protection for the joint release of an output from \(M\) and the invariant \(t\) as follows:

\begin{proposition}\label{thm: full characterize}(Full Characterization)
Given a confidential dataset \(X^*\), let the invariant be \(T(X^*) = t\) and let $M:\mathcal{D} \rightarrow \mathcal{Y}$ be a $(\mathcal{D},A,f)$-DP mechanism, where an adjacency function $A: \mathcal{D} \times \mathcal{D} \rightarrow \{0,1\}$ is defined via an adjacency metric by $A(X,X') = \mathbbm{1}_{\{d(X,X') \leq 1\}}$. Consider invariant-conforming databases $\mathcal{D}_{t} = \{X \in \mathcal{D}: T(X)=t\}$. Then for any invariant-conforming databases $X,X' \in \mathcal{D}_{t}$, we have 
\begin{equation*}
    T\left(M(X),M(X')\right) \geq f^{\circ d(X,X')}.
\end{equation*}
\end{proposition}

Proposition \ref{thm: full characterize} applies to all data pairs and provides a tight lower bound on the trade-off function without knowing further details of the privacy mechanism.  Proposition \ref{thm: full characterize} follows from the observation that $M$ is $f$-DP if and only if $M$ is $f^{\circ d(X,X')}$ for all $X,X' \in \mathcal{D}$, given that $d$ is an adjacency metric and $\mathcal{D}_{t} \subseteq \mathcal{D}$. 

However, it is challenging to understand what this characterization means in practical terms and what kind of privacy protection it offers from the perspective of individuals whose data is in the database. In this regard, we introduce our primary framework, semi DP, designed to preserve the original spirit of DP while accounting for the presence of invariants.

\begin{definition}\label{def: semi dp} ($f$-semi differential privacy).
Given a confidential dataset \(X^*\), let the invariant be \(T(X^*) = t\) and $M(X)$ be a privacy mechanism. For the set of invariant-conforming databases \(\mathcal{D}_{t} = \{X \in \mathcal{D} : T(X) = t\}\), and given an adjacency function \(A: \mathcal{D}_{t} \times \mathcal{D}_{t} \rightarrow \{0,1\}\), a joint release of \((M(X), T(X))\) is said to satisfy \(f\)-semi DP if \(M(X)\) satisfies \((\mathcal{D}_{t}, A, f)\)-DP. That is, 
\[
T(M(X), M(X')) \geq f,
\]
for all \(X, X' \in \mathcal{D}_{t}\) such that \(A(X, X') = 1\).
\end{definition}

The primary distinction between our definition of \(f\)-semi DP and the conventional \(f\)-DP in Definition \ref{def: conventional f dp} lies in its explicit restriction to invariant-conforming database, rather than the original dataspace. Thus, Semi-DP guarantees indistinguishability among adjacent databases within the restricted dataspace \(\mathcal{D}_{t}\).

To properly define adjacency within this restricted space, we introduce the semi-adjacent parameter  \( a(t) \) and consider an adjacency function of the form \( A_{a(t)}(X, X') = \mathbbm{1}_{\{d(X, X') \leq a(t)\}} \) . The two most important considerations in defining the semi-adjacent parameter \( a(t) \) are: ensuring that every \( X \in \mathcal{D}_{t} \) has an adjacent counterpart \( X' \), and providing a meaningful interpretation of individual-level protection, consistent with the original differential privacy framework.

\begin{definition}\label{def: indiv protection a} (Semi-adjacent parameter).
Let \( \mathcal{D}_t \) be a collection of datasets, each consisting of \( n \) individuals, and let \( \mathcal{D}_t^i = \{x \in \mathcal{X} : \exists X \in \mathcal{D}_t, X_i = x \} \) represent the set of all possible data values for individual \( i \) across datasets in \( \mathcal{D}_t \). Define the semi-adjacent parameter as:
\[
a(t) = \sup_{i \in [n]} \sup_{\substack{x, y \in \mathcal{D}_t^i}} \inf \{ d(X, Y) : X, Y \in \mathcal{D}_t, X_i = x, Y_i = y \},
\]
where \( d(X, Y) \) is an adjacent metric between datasets \( X \) and \( Y \).
\end{definition}

The semi-adjacent parameter \( a(t) \) quantifies the worst-case scenario for replacing an individual \( x \) in the dataset with another individual \( y \), ensuring that the two resulting databases \( X \) and \( Y \) maintain the same invariant values. Importantly, this is not just about finding \textit{any} replacement \( y \), but identifying the most difficult or ``worst-case" replacement. By defining the adjacency function as \( A_{a(t)}(X, X') \), we ensure that, for each individual \( x \), even the most challenging alternative \( y \) is still indistinguishable from the perspective of the adversary. This focus on the worst-case replacement is crucial for ensuring that privacy holds under the most stringent conditions, consistent with the principles of DP.

\begin{remark}
One may consider other options for the semi-adjacent parameter. 
It is worth noting that the seemingly intuitive choice of \( a_{min} := \min_{X, X' \in \mathcal{D}_{t}} d(X, X') \) is inappropriate. This is because a particular confidential dataset \(X^*\) may not have any adjacent counterpart within a distance of \(a_{min}\), i.e., \(\{X \in \mathcal{D}_{t} \mid d(X, X^*) \leq a_{min}\} = \{X^*\}\). This situation arises if \(X^*\) is not one of the databases achieving \(a_{min}\), rendering \(a_{min}\) unsuitable for DP.

Alternatively, one might consider \( a_{maxmin}(t) := \max_{X \in \mathcal{D}_{t}} \min_{X' \in \mathcal{D}_{t}} d(X, X') \), which reflects the worst-case scenario by focusing on the maximum distance to the nearest database in \(\mathcal{D}_{t}\). While \(a_{maxmin}\) ensures that every database has an adjacent counterpart, it lacks a clear interpretation regarding individual-level privacy. For example, $d(X,X') \leq a_{maxmin}$ may only allow a record to differ in one attribute rather than arbitrarily swapping the record for another.

Therefore, we adopt \(a_{t}\) as defined in Definition \ref{def: indiv protection a}, which guarantees that every dataset in \(\mathcal{D}_{t}\) has an adjacent counterpart while allowing us to reason explicitly about individual privacy protection.    
\end{remark}

\begin{example}
    Suppose $\mathcal{D} = \{(x_{1},x_{2},x_{3}):x_{i} \in \{0,1\}$, \text{for all} $i=1,2,3\}$ and consider $T(X) = \#\{i : x_{i} = 1\}$ as defined in Example \ref{example: motivation}. 

   If \(t_{1} = 2\), then the set of invariant-conforming databases is \(\mathcal{D}_{t_1} = \{(1,1,0), (1,0,1), (0,1,1)\}\), so \(\mathcal{D}_{t_1}^{i} = \{0,1\}\) for each \(i \in \{1,2,3\}\). The semi-adjacency parameter in this case is \(a(t_1) = 2\), indicating that two switches in entries are required to transition between any two databases in \(\mathcal{D}_{t_1}\) while preserving the invariant.

If \(t_{2} = 0\), then \(\mathcal{D}_{t_2} = \{(0,0,0)\}\), so \(\mathcal{D}_{t_2}^{i} = \{0\}\) for each \(i \in \{1,2,3\}\). In this case, since there is only a single database conforming to the invariant \(t_{2}\), the semi-adjacency parameter is \(a(t_2) = 0\), as there is nothing left to protect.
\end{example}

We conclude this subsection by presenting the result on how the privacy guarantee of a mechanism \( M(X) \), which satisfies \((\mathcal{D}, A_{1}, f)\)-DP, is affected by the invariant \( t \) through the use of the semi-adjacent parameter \( a(t) \).

\begin{proposition}\label{thm: semi dp by semi adj para} Given a confidential dataset \(X^*\), let the invariant be \(T(X^*) = t\) and $a(t)$ be the semi-adjacent parameter. Let $M:\mathcal{D} \rightarrow \mathcal{Y}$ be a $(\mathcal{D},A_{1},f)$-DP mechanism, where an adjacency function $A: \mathcal{D} \times \mathcal{D} \rightarrow \{0,1\}$ is defined via an adjacency metric by $A(X,X') = \mathbbm{1}_{\{d(X,X') \leq 1\}}$. Then $M(X)$ satisfies $(\mathcal{D}_{t},A_{a(t)},f^{\circ a(t)})$-DP.
\end{proposition}

Proposition \ref{thm: semi dp by semi adj para} demonstrates that when a mechanism \(M\) satisfies \((\mathcal{D}, A_{1}, f)\)-DP, the introduction of an invariant \(t\) modifies the privacy guarantee by restricting attention to invariant-conforming databases. Specifically, instead of ensuring indistinguishability between all adjacent databases within \(\mathcal{D}\) at distance 1, the privacy guarantee is on indistinguishability between databases in \(\mathcal{D}_t\) that differ by up to \(a(t)\), where \(a(t)\) is the semi-adjacent parameter. The privacy parameter is adjusted from \(f\) to \(f^{\circ a(t)}\), reflecting the privacy loss over the greater distance between adjacent databases. As a result, the privacy guarantee is weakened both by reducing the set of protected databases and by tightening the privacy parameter.

It is important to note that while this result appears similar to group privacy with group size \(a = a(t)\), there is a key distinction. In group privacy, the privacy guarantee applies across all the pairs of databases differing by up to $a$ entries. However, in this case, the privacy guarantee is limited to databases with up to $a$ different entries within \(\mathcal{D}_t\). Therefore, even though the privacy guarantee in Proposition \ref{thm: semi dp by semi adj para} and group privacy both involve the privacy parameter \(f^{\circ a}\), the guarantee in Proposition \ref{thm: semi dp by semi adj para} should be understood as weaker than group privacy because it is constrained to invariant-conforming databases, rather than applying broadly across all possible databases. This idea is captured and formalized later in Definition \ref{def: partial order}.

Finally, this subsection focused on how the privacy guarantee of a \((\mathcal{D}, A, f)\)-DP mechanism \(M\) changes when an invariant \(t\) is introduced. However, it is also possible to design privacy mechanisms that account for the invariant from the outset. In Section \ref{sec: mechanism design}, we will introduce a customized mechanism that incorporates the invariant directly into the mechanism design while achieving the same Semi-DP guarantee. 

\subsection{Example on Count Statistics}\label{sec: example on count statistics}

In this section, we calculate the semi-adjacent parameter $a(t)$ for some examples involving count statistics. 

For a set of variables \(\{A_1, \dots, A_p\} \), where each variable \( A_j\) can take \( n_j \) distinct levels, consider a dataset \( X = (X^{(1)}, \dots, X^{(n)}) \in \mathcal{X}^n \), where each data point \( X^{(i)} = (X_1^{(i)}, \dots, X_p^{(i)}) \) consists of \( p \) features corresponding to the variables \( A_1, \dots, A_p \). 

For each variable \( A_j \), we define the count statistic \( T_j(X) \) as the ordered tuple of the counts of each level \( k \in \{1, \dots, n_j\} \), given by:
\[
T_j(X) := \left( \#\{i : X_j^{(i)} = k\} \right)_{k=1}^{n_j},
\]
where \( \#\{i : X_j^{(i)} = k\} \) represents the number of data points indexed over \(i\), for which the \( j \)-th feature takes the value \( k \).

To capture the counts of all features in the dataset, we introduce the notation \( \mathbf{T}(X) \), which represents the tuple of count functions for all \( p \) features:
\[
\mathbf{T}(X) = \left( T_1(X), T_2(X), \dots, T_p(X) \right),
\]
where each \( T_j(X) \) provides a count of the levels for the corresponding variable \( A_j \).

\begin{example}(One-way margins of $2 \times 2$ contingency table).
Consider two binary variables \( A_1 \) and \( A_2 \), where each variable takes values in \( \{1, 2\} \). Let \( X = (X^{(1)}, \dots, X^{(n)}) \) be a dataset consisting of \( n \) data points, and suppose the corresponding contingency table for \( A_1 \) and \( A_2 \) is given as follows:

\[
\begin{array}{c|cc}
A_1 \backslash A_2 & 1 & 2 \\
\hline
1 & x_{11} & x_{12} \\
2 & x_{21} & x_{22}
\end{array}
\]

The one-way margins for \( A_1 \) represent the total counts for each value of \( A_1 \), which are \( x_{11} + x_{12} \) and \( x_{21} + x_{22} \). Similarly, the one-way margins for \( A_2 \) are given by \( x_{11} + x_{21} \) and \( x_{12} + x_{22} \).

In our notation, \( \mathbf{T}(X) = (T_1(X), T_2(X)) \), where \( T_1(X) \) gives the counts of the levels of \( A_1 \), and \( T_2(X) \) gives the counts of the levels of \( A_2 \). Specifically, we observe:
\[
T_1(X) = \left( \#\{i: X_1^{(i)} = 1\}, \#\{i: X_1^{(i)} = 2\} \right) = \left( x_{11} + x_{12}, x_{21} + x_{22} \right),
\]
\[
T_2(X) = \left( \#\{i: X_2^{(i)} = 1\}, \#\{i: X_2^{(i)} = 2\} \right) = \left( x_{11} + x_{21}, x_{12} + x_{22} \right).
\]

Thus, \( \mathbf{T}(X) \) coincides with the one-way margins of the contingency table, demonstrating that the notation captures the count of individual features.
\end{example}

Suppose the invariant is $t = \mathbf{T}(X^*)$, where $X^*$ is the confidential database. Then we have bounds for the semi-adjacent parameter $a(t)$ as follows:

\begin{theorem}\label{thm: count stat semi adj para}
Let \( X^* \) be the confidential database, and let the one-way margins \( t = \mathbf{T}(X^*) \) be the invariant. Consider the Hamming distance as the adjacency metric. Then the semi-adjacent parameter $a(t)$ satisfies the following bound:
\[
a(t) \leq p + 1.
\]
\end{theorem}

The upper bound of \( p + 1 \) arises when transforming an individual \( x \) into another individual \( y \), where all \( p \) features differ, and no individual in the dataset shares more than one feature with \( y \). Under these conditions, each of the \( p \) features requires a distinct replacement operation, leading to \( p \) replacements. Including the initial replacement of \( x \), this results in a total of \( p + 1 \) replacements. Note that there are cases where $a(t)<p+1$, such as when some of the marginal counts are zero.

\begin{example}\label{example: 2 by 2 semi adj para}(Semi-adjacent parameter given one-way margins of a \( 2 \times 2 \) contingency table).

Consider a \( 2 \times 2 \) contingency table, where \( p = 2 \), so \( a(t) \leq 3 \). To illustrate a scenario where \( a(t) = 3 \), let \( x \) be an individual with feature values \( A_1 = 1 \) and \( A_2 = 1 \). Suppose \( x \), observing that the one-way margins \( x_{21} + x_{22} \) and \( x_{12} + x_{22} \) are nonzero, attempts to pretend to be an individual with feature values \( A_1 = 2 \) and \( A_2 = 2 \). The worst-case scenario occurs when \( x_{22} = 0 \), meaning there are no individuals in the dataset with feature $A_{1}=2,A_{2}=2$. Note that for the one-way margins to remain nonzero, \( x_{21} \) and \( x_{12} \) must be positive. Thus, we can select an individual with \(A_2 = 2\) and \(A_1 = 1\), and another individual with \(A_1 = 1\) and \(A_2 = 1\). By replacing these individuals with others who also have \(A_1 = 1\) and \(A_2 = 1\), we can preserve the original counts. Thus, in addition to the initial change of \( x \), two further changes are required to maintain the counts, resulting in \( a(t) = 3 \). The resulting changes in terms of the contingency table can be expressed as \(x_{11} + 1\), \(x_{12} - 1\), \(x_{21} - 1\), and \(x_{22} + 1\).
\end{example}

\begin{remark}
An interesting extension of our result arises when more detailed counts are to be released instead of the marginal counts for each of the \( p \) features \( A_1, \dots, A_p \). For example, consider a dataset of college students where \( A_1 \) represents gender, \( A_2 \) represents major, and \( A_3 \) represents state of residence. In the case of \( \mathbf{T}(X) \), marginal counts for gender, major, and state would be published separately. However, a data curator may wish to release joint counts for gender and major, which correspond to counts such as (male, statistics major), (female, statistics major), (male, CS major), (female, CS major), and so on. This provides more detailed information than marginal counts alone. Conceptually, we can view this as combining \( A_1 \) and \( A_2 \) into a single feature \( A_1 \times A_2 \), effectively reducing the number of features to \( p - 1 \). According to Theorem \ref{thm: count stat semi adj para}, the semi-adjacent parameter \( a(t) \) would then be bounded by \( p \).

At first, this might seem counterintuitive: publishing more detailed information results in a smaller \( a(t) \). However, the key observation is that as more detailed information is released, the set of datasets consistent with the invariant becomes smaller. This reduction effectively constrains the set of potential replacements for an individual \( x \) with another individual \( y \). As a result, the number of viable \( y \) values to consider when evaluating the semi-adjacent parameter \( a(t) \) decreases. Since the supremum over \( x, y \in \mathcal{D}_{t}^{i} \) is taken over a smaller set, the upper bound of \( a(t) \) naturally decreases.

When the counts for every level in $A_{1} \times \cdots \times A_{p}$ are published, \( x \) must select one of the non-zero count feature combinations and wishes to pretend to be an individual \( y \) with those feature values. Since there is guaranteed to be an individual in the dataset with the chosen feature combination, changing that individual to one with \( x \)'s original feature values preserves the counts. As a result, only two replacements are required.

Thus, as more detailed information is published, the selection of a plausible \( y \) becomes more constrained, leading to the somewhat surprising but logical result that the upper bound of \( a(t) \) decreases. Note that in the extreme case where $t$ reveals the entire dataset, we have $a(t)=0$.
\end{remark}

Finally, we have the following privacy guarantee for a joint release of such invariant $t$ and a DP mechanism: 

\begin{corollary}\label{thm: count stat semi dp guarantee}
Given a confidential database $X^*$ that consists of $p$ features, let the invariant be \( t = \mathbf{T}(X^*) \) and $M$ be a pivacy mechansim that satisfies $(\mathcal{D},A_{1},f)$-DP. Then the joint release $(M(X),t)$ satisfies $(\mathcal{D}_{t},A_{a(t)},f^{\circ (p+1)})$-DP.
\end{corollary}

Corollary \ref{thm: count stat semi dp guarantee} is immediate from Theorem \ref{thm: count stat semi adj para}, Proposition \ref{thm: semi dp by semi adj para}, and the fact that $f^{\circ (p+1)} \leq f^{\circ a(t)}$, which is because $a(t)\leq p+1$ and $f(x) \leq x$.

\subsection{Comparison Between Semi-DP Guarantees}

In our Semi-DP framework, comparing guarantees is inherently more challenging than in standard DP. For standard DP, comparisons are simplified because the guarantees share a common underlying dataspace and adjacency function, making it sufficient to compare the privacy parameters alone. However, in Semi-DP, the underlying dataspace and corresponding adjacency function vary based on the invariant \( t \). Therefore, relying solely on privacy parameters for comparison can be misleading. 

\begin{example}\label{example: group vs. semi}(Comparing group privacy and Semi-DP). To illustrate the case where merely comparing privacy parameters could be misleading, consider a $(\mathcal{D},A,f)$-DP mechanism $M(X)$ and invariant $t$ and suppose $a(t)=a$. On one hand, $M(X)$ satisfies $(\mathcal{D},A_{a},f^{\circ a})$-DP as the group privacy of size $a$. On the other hand, in $\mathcal{D}_{t}$, for $X,X' \in \mathcal{D}_{t}$ such that $A_{a}(X,X')=1$, their tradeoff function has a lower bound as $f^{\circ a}$ by Proposition \ref{thm: full characterize}, leading to $(\mathcal{D}_{t},A_{a},f^{\circ a})$-DP, or $f^{\circ a}$-semi DP.

If we only compare privacy parameters, the two privacy guarantees should have the same level of protection. However, this is incorrect. Group privacy quantifies the indistinguishability for all data pairs in the dataspace with a size difference of $a$, while Semi-DP quantifies the indistinguishability for some data pairs in the  $\mathcal{D}$ with a size difference of $a$. Roughly speaking, some data pairs protected under group privacy are not protected under Semi-DP. 
\end{example}

To better understand what pairs of databases are protected, we define the ``indistinguishable pairs" of a privacy guarantee.

\begin{definition}\label{def: indistinguishable pairs}(Indistinguishable pairs). Given an underlying dataspace $\mathcal{D}$ and an adjacency function $A: \mathcal{D} \times \mathcal{D} \rightarrow \{0,1\}$, we define the indistinguishable pairs as the collection of pairs of databases which are adjacent:
\begin{equation*}
    IND(\mathcal{D},A) = \{(X,X') \in \mathcal{D}\times \mathcal{D}: A(X,X')=1\}.
\end{equation*}
\end{definition}

Note that Definition \ref{def: indistinguishable pairs} outlines the pairs of databases that can be protected by a given privacy guarantee, effectively defining the scope of the privacy protection. For example, a standard DP mechanism always has $IND(\mathcal{D},A_1)$ since all adjacent databases are of interest. On the other hand, for group privacy of size $a$, the corresponding $IND$ is $IND(\mathcal{D},A_a)$ because group privacy targets indistinguishability between any databases with a size difference of $a$. However, in Semi-DP, the $IND$ changes depending on invariant.

\begin{definition}\label{def: partial order}
    Consider two DP guarantees, $(\mathcal{D}, A, f)$-DP and $(\mathcal{D}', A', g)$-DP. We define a partial order $\preceq$ between these guarantees as follows:
    
    \[
    (\mathcal{D}', A', g)\text{-DP} \preceq (\mathcal{D}, A, f)\text{-DP} \quad \text{if:}
    \]
    \begin{align*}
        1. &\quad \text{IND}(\mathcal{D}', A') \subseteq \text{IND}(\mathcal{D}, A) \text{ and } \\
        2. &\quad g \leq f.
    \end{align*}

    We say that $(\mathcal{D}, A, f)$-DP is a stronger privacy guarantee than $(\mathcal{D}', A', g)$-DP if 
    \[
    (\mathcal{D}', A', g)\text{-DP} \preceq (\mathcal{D}, A, f)\text{-DP}.
    \]

    Furthermore, we define $(\mathcal{D}', A', g)$-DP $\prec$ $(\mathcal{D}, A, f)$-DP if one of the following conditions holds:
    \begin{align*}
        1. &\quad \text{IND}(\mathcal{D}', A') \subseteq \text{IND}(\mathcal{D}, A) \text{ and } g < f, \text{ or } \\
        2. &\quad \text{IND}(\mathcal{D}', A') \subset \text{IND}(\mathcal{D}, A) \text{ and } g \leq f.
    \end{align*}
\end{definition}

In this definition, the partial order \(\preceq\) allows us to compare DP guarantees based on their indistinguishable pairs and tradeoff functions, while the strict partial order \(\prec\) indicates that one guarantee is strictly stronger if it improves on one of the criteria while not being weaker in the other.

Definition \ref{def: partial order} highlights the importance of considering both the indistinguishable pairs and the privacy parameter to evaluate a privacy guarantee. Recall Example \ref{example: group vs. semi}, which compares \((\mathcal{D}, A_{a}, f^{\circ a})\)-DP and \((\mathcal{D}_{t}, A_{a}, f^{\circ a})\)-DP. Although they share the same lower bound on the tradeoff function \(f^{\circ a}\), we have \(\text{IND}(\mathcal{D}_{t},A_a) \subseteq \text{IND}(\mathcal{D},A_{a})\), implying \((\mathcal{D}_{t}, A_{a}, f^{\circ a})\)-DP \(\preceq\) \((\mathcal{D}, A_{a}, f^{\circ a})\)-DP.

\subsection{Properties of Semi-DP}
The Semi-DP framework retains many characteristics of DP within the restricted dataspace. \\

\noindent{\bf Composition:} When a data curator considers a joint release of composition among DP mechanisms and invariants, there are two options in analyzing the privacy guarantee: calculate Semi-DP first and then apply composition (Semi-DP first), or applying composition and then calculating Semi-DP (composition first). Theorem \ref{thm: semi dp composition} tells us that employing Semi-DP first offers a stronger privacy guarantee.

\begin{proposition}\label{thm: semi dp composition} Suppose $M_{i}: \mathcal{D} \rightarrow \mathcal{Y}$ satisfies $(\mathcal{D},A_{1},f_{i})$-DP for all $i \in [k]$. Let $a(t)$ be the semi-adjacent parameter given the invariant $t$. Then Semi-DP first satisfies $(\mathcal{D}_{t},A_{a(t)},f_{1}^{\circ a} \otimes \cdots \otimes f_{k}^{\circ a})$-DP and Composition first satisfies $(\mathcal{D}_{t},A_{a(t)},(f_{1} \otimes \cdots \otimes f_{k})^{\circ a})$-DP.

Moreover, $(\mathcal{D}_{t},A_{a(t)},(f_{1} \otimes \cdots \otimes f_{k})^{\circ a})$-DP $\preceq$ $(\mathcal{D}_{t},A_{a(t)},f_{1}^{\circ a} \otimes \cdots \otimes f_{k}^{\circ a})$-DP.
\end{proposition}

Note that both approaches share the same indistinguishable pairs. Therefore, comparison between the two privacy guarantees depends only on the tradeoff function. We establish that the tradeoff function for Semi-DP has a lower bound of \( f_{1}^{\circ a} \otimes \cdots \otimes f_{k}^{\circ a} \), while the Composition approach yields \( (f_{1} \otimes \cdots \otimes f_{k})^{\circ a} \). Using the result from \cite{awan2022log} that we present as Lemma \ref{lemma: tradeoff function} in the Appendix, we have that \( (f_{1} \otimes \cdots \otimes f_{k})^{\circ a} \) is the smaller function. Therefore, Theorem \ref{thm: semi dp composition} implies that Semi-DP first provides a stronger privacy guarantee.

\begin{remark}
The composition of $k$ Semi-DP mechanisms $((M_{1}(X),T_{1}(X)), \cdots, (M_{k}(X),T_{k}(X)))$ can be addressed by Theorem \ref{thm: semi dp composition} by taking $T(X) = (T_{1}(X),\cdots,T_{k}(X))$. On the other hand, we do not have an efficient method to sequentially update the impact due to invariant. 
\end{remark}

\noindent
{\bf Invariance to Post-Processing:} Similar to DP, in Semi-DP, applying further data-independent transformations to the mechanism's output will not weaken the privacy guarantee. Note that in Semi-DP, the transformation may depend on the invariant.

\begin{proposition}\label{thm: semi dp postprocessing}
    Suppose $M(X)$ is $(\mathcal{D}_{t},A,f)$-DP, and let $\proc$ be a randomized function taking as input the output of $M(X)$ and the invariant $t$. Then for any pair of databases $X, X' \in \mathcal{D}_{t}$ such that $A(X,X')=1$, we have
    \begin{equation*}
        T(\proc(M(X),t),\proc(M(X')),t) \geq f.
    \end{equation*}
\end{proposition}

\noindent
{\bf Converting to Other DP Guarantees:} In Section \ref{sec: preliminary}, we discussed the conversion between \(\mu\)-GDP and \((\epsilon,\delta)\)-DP. Since our Semi-DP framework relies on the same similarity measures as in existing DP definitions while focusing on the dataspace and adjacency relations influenced by invariants, all conversions valid in standard DP should remain applicable in our framework. For example, $f$-DP can be converted to divergence-based version DP, such as R\'enyi-DP \citep{mironov2017renyi}.

\section{Mechanism Design}\label{sec: mechanism design}

In this section, we introduce customized mechanisms which leverage the invariant in their construction in order to minimize the amount of noise added to achieve Semi-DP. To this end, we leverage the sensitivity space, which is a fundamental tool in additive mechanisms. Specifically we develop the optimal $K$-norm mechanism in Semi-DP.

\subsection{Mechanism Design via Analysis on Sensitivity Space}

If Semi-DP is the target privacy guarantee, we can optimize the choice of the mechanism \( M \) and any  post-processing to leverage the invariant \( t \). By incorporating \( t \), it is possible to reduce unnecessary noise while still satisfying the target privacy guarantee.

Recall that the sensitivity space, as defined in Definition \ref{def: sensitivity space}, encompasses all possible differences in a query $\phi: \mathcal{D} \rightarrow \mathbb{R}^{d}$ between adjacent databases. The additive noise is scaled according to the sensitivity of $\phi$, which is the largest radius of this sensitivity spaceto mask the impact of differences in adjacent input databases.

In Semi-DP given invariant $t$ and target privacy guarantee $(\mathcal{D}_{t},A_{a(t)},f)$, the sensitivity space of $\phi$ is 
\[
S_{Semi} = \{\phi(X)-\phi(X'): A_{a(t)}(X,X')=1\}.
\]

Using the Semi-DP sensitivity space, we can seamlessly apply many existing DP methods. This encompasses traditional mechanisms like the Gaussian and Laplace mechanisms, as well as more advanced methods such as DP-empirical risk minimization \citep{chaudhuri2011differentially} and DP-stochastic gradient descent \citep{abadi2016deep}. Consequently, our framework supports the flexible application of established DP tools, distinguishing itself from prior work \citep{gong2020congenial, gao2022subspace}, which focused solely on invariants which are linear transformations and mechanisms of a particular structure.

While the idea of using \( S_{Semi} \) is simple, it is also highly effective. This effectiveness becomes evident when compared to other naive approaches. A straightforward method for designing a mechanism that satisfies \((\mathcal{D}_{t}, A_{a(t)}, f)\)-DP is to first construct a mechanism that satisfies \((\mathcal{D}, A_{1}, g)\)-DP, where \( g^{\circ a} \geq f \). By leveraging group privacy, this mechanism would then satisfy \((\mathcal{D}, A_{a}, f)\)-DP, and consequently \((\mathcal{D}_{t}, A_{a(t)}, f)\)-DP as well. Another naive approach would be to directly design a mechanism that satisfies \((\mathcal{D}, A_{a}, f)\)-DP, that is group privacy of size $a$ with privacy parameter $f$, which then automatically satisfies \((\mathcal{D}_{t}, A_{a(t)}, f)\)-DP. However, neither of thses approaches fully leverages the information in $t$.

To further illustrate the advantage of leveraging $S_{Semi}$ from the perspective of the sensitivity space, let us consider the Gaussian mechanism that satisfies \(\mu\)-GDP. In the first naive approach, we use the sensitivity calculated from \( S_{DP} := \{\phi(X) - \phi(X'): A_{1}(X, X') = 1\} \) to construct a mechanism that satisfies \((\mathcal{D}, A_{1}, G_{\mu/a})\)-DP. According to the Gaussian mechanism described in Proposition \ref{prop: standard gaussian mechanism}, substituting \(\mu\) with \(\mu/a\) adjusts the variance of the noisy Gaussian noise to $\sigma \geq $\(a \Delta_{2} / \mu\). Thus, the first method can be interpreted as calibrating noise using an enlarged sensitivity space \( a S_{DP} \), effectively scaling the original \( S_{DP} \) by \( a \). On the other hand, the second approach directly targets indistinguishability for all pairs within distance \( a \) in the dataspace, leading to a sensitivity space defined as \( S_{Group} := \{\phi(X) - \phi(X'): A_{a}(X, X') = 1\} \). Comparing \( S_{Semi} \) with \( a S_{DP} \) and \( S_{Group} \), we the following:

\begin{proposition}\label{thm: naive SS large}
For $S_{DP}$, $S_{Group}$ and $S_{Semi}$, we have 
\begin{equation*}
S_{Semi} \subseteq S_{Group} \text{ and } \hull(S_{Group}) \subseteq a\hull(S_{DP}).
\end{equation*}
Moreover, $\sup_{u \in S_{Semi}}\Vert u \Vert \leq \sup_{u \in S_{Group}}\Vert u \Vert \leq a \sup_{u \in S_{DP}}\Vert u \Vert$,
for any norm $\Vert \cdot \Vert$.
\end{proposition}

According to Proposition \ref{thm: naive SS large}, we conclude that \( S_{Semi} \) allows smaller noise than using other naive counterparts. This is because \( S_{Semi} \) requires masking only the differences that occur between pairs we target to protect, instead of considering differences across all pairs, including those that cannot be protected due to the given invariant $t$. \\

\noindent
{\bf Downsizing the Noise by Projection:} We propose a method to reduce the additive noise 
when the sensitivity space is rank-deficient, meaning rank$(\Span(S)) < d$. A typical example occurs when the invariant $t$ is a linear transformation of the query $\phi$, as considered in \cite{gao2022subspace}.

\begin{proposition}\label{prop: linear invariant rank def}
    Suppose the invariant \( t \) is a non-trivial linear transformation of \( \phi(X) \), meaning \( t = L\phi(X) \) for some matrix \( L \) such that \( \text{ker}(L) \neq \mathcal{D} \), and the semi-adjacent parameter \( a(t) \) is given. Then, the sensitivity space \( S_{Semi} = \{\phi(X) - \phi(X') : A_{a(t)}(X,X') = 1 \text{ and } L\phi(X) = L\phi(X')\} \) is rank-deficient.
\end{proposition}

In such cases, we do not need to add noise in these negligible directions as adjacent databases yield identical values, and so they provide no clues to the adversary for distinguishing between the adjacent databases. By projecting the noise onto the subspace where the query value actually changes, we can avoid adding noise to such directions, leading to the same privacy guarantee with smaller noise.

\begin{proposition}\label{thm: projection framework} (Projection to Downsize Additive Noise) Let $\phi:\mathcal{D} \rightarrow \mathbb{R}^{d}$ be a query, and suppose the mechanism $M(X) = \phi(X) + e$ satisfies $(\mathcal{D}, A, f)$-DP, where $e$ is a random vector. Then, a mechanism $\phi(X) + \Proj_{\mathcal{S}}e$ also satisfies $(\mathcal{D}, A, f)$-DP, where $\mathcal{S} = \Span(S_{\phi})$ is the span of the sensitivity space and $\Proj_{\mathcal{S}}$ is the orthogonal projection operator onto $\mathcal{S}$. \end{proposition}

In Proposition \ref{thm: projection framework}, if the original noise \(e\) perturbs all \(d\) dimensions in which the query takes values, then \(\Proj_{\mathcal{S}} e\) serves to add noise only in the directions where changes occur between adjacent databases. This ensures that noise is applied effectively, targeting only the relevant dimensions influenced by the data change. Indeed, it follows from standard linear algebra that$\Vert \Proj_{\mathcal{S}}e \Vert \leq \Vert e \Vert$ for any norm, as projecting onto a subspace cannot increase the norm, where the strong inequality holds when $\mathcal{S}$ is rank-deficient. Therefore, Proposition \ref{thm: projection framework} helps us to downsize the noise. 

\begin{remark}
Proposition \ref{thm: projection framework} is especially useful when the invariant is a linear transformation of the query $\phi$. For instance, the 2020 US Census has invariant $T(X)$ as a linear mapping of true counts on demographics $\phi(X)$, say $T(X) = C\phi(X)$ for some matrix $C$. Then $T(X)=t$ allows us not to add noise onto subspace $\Span{(C)}$ and in this case, it agrees with the mechanism design in \cite{gao2022subspace} for subspce DP. Note that a similar projection idea was introduced in \cite{kim2022differentially} for general queries, but without considering invariants, and with a specific focus on the Gaussian mechanism. However, we emphasize that Theorem \ref{thm: projection framework} presents a general concept applicable to all additive noise mechanisms. 
\end{remark}

We introduce a modified Gaussian mechanism that satisfies \((\mathcal{D}_t, A_{a(t)}, G_{\mu})\)-DP by leveraging the sensitivity space \( S_{Semi} \) and applying the projection technique outlined in Proposition \ref{thm: projection framework}.

\begin{algorithm}[htb!]
\caption{\texttt{Gaussian Mechanism}}
\label{alg: gaussian mechanism}
\begin{algorithmic}[1]
    \STATE \textbf{Input:} Query $\phi:\mathcal{D}\rightarrow \mathbb{R}^{d}$ to be privatized, invariant $t$, privacy parameter $\mu$
    \STATE Get $S_{Semi} = \{\phi(X) - \phi(X'): A_{a(t)}(X,X')=1\}$ 
    \STATE Define $P$ as an orthogonal projection matrix onto $\Span{(S_{Semi})}$
    \STATE Calculate $\ell_{2}$-sensitivity $\Delta_{2}(S_{Semi}) = \sup_{u \in S_{Semi}} \Vert u \Vert_{2}$
    \STATE Sample $N \sim N_{d}\left(0, (\Delta_{2}(S_{Semi})/\mu)^{2}P\right)$
    \STATE \textbf{Output: } $M(X) = \phi(X) + N$
\end{algorithmic}    
\end{algorithm}

Algorithm \ref{alg: gaussian mechanism} outlines the Gaussian mechanism. The noise \( N \) in Algorithm \ref{alg: gaussian mechanism} can equivalently be understood as the result of sampling from \( N_{d}(0, (\Delta_{2}(S_{Semi})/\mu)^{2} I_{d}) \) and then projecting onto \(\Span(S_{Semi})\). In this context, the projection matrix \( P \) projects vectors from the ambient space \(\mathbb{R}^d\) onto the subspace \(\Span(S_{Semi})\). It indicates that \( P \) effectively reduces the variance of the noise by confining it to the subspace \(\Span(S_{Semi})\). Notably, if \(\{u_{1}, \cdots, u_{s}\}\) is an orthonormal basis of \(\Span(S_{Semi})\), where \( s \leq d \), the projection matrix \( P \) can be constructed as \( P = \sum_{i=1}^{s} u_{i}u_{i}^\top \).

Proposition \ref{thm: gaussian mech guarantee} establishes the privacy guarantee of the Gaussian mechanism.

\begin{proposition}\label{thm: gaussian mech guarantee}
    The mechanism \( M(X) \) in Algorithm \ref{alg: gaussian mechanism} satisfies \((\mathcal{D}_t, A_{a(t)}, G_{\mu})\)-DP.
\end{proposition}

\subsection{Optimal $K$-Norm Mechanism}\label{sec: optimal K norm mech}

$K$-norm mechanisms are a multi-dimensional extension of the Laplace mechanism that satisfy $\epsilon$-DP, or equivalently, $f_{\epsilon,0}$-DP. Unlike the Gaussian mechanism, one has the freedom to customize the norm. In \cite{awan2021structure}, they proved that a norm ball $K$ which is the convex hull of the sensitivity space leads to the \emph{optimal} $K$-norm mechanism, given that the underlying sensitivity space is of full-rank. 

However, the introduction of an invariant $t$ can make the sensitivity space rank-deficient, rendering their strategy no longer valid as the convex hull of a rank-deficient sensitivity space is not guaranteed to be a norm ball in $\mathbb{R}^{d}$ as it is not of full dimension. To address this, we consider the convex hull of the sensitivity space in the subspace spanned by the sensitivity space, where the convex hull becomes a valid norm ball.

\begin{lemma}\label{thm: convex hull norm ball}
    Let $K$ = $\hull(S)$ be the convex hull of $S$. Then $K$ is a norm ball in $\mathcal{S}=\Span(S)$, provided that the sensitivity space $S$ is bounded.
\end{lemma}

Note that the boundedness condition on the sensitivity space is inevitable, otherwise, we cannot mask the change between adjacent databases. For $K = \hull(S)$, we define the \( K \)-norm  \(\Vert \cdot \Vert_{K}: \mathcal{S} \rightarrow \mathbb{R}^{\geq 0}\) in the subspace \(\mathcal{S}\) by $\Vert v \Vert_{K} := \inf\{c \in \mathbb{R}^{\geq 0}: u \in cK \}$. If $s$ is the dimension of $\mathcal{S}$, we can further define a corresponding norm in \(\mathbb{R}^s\), an isomorphic counterpart of $\mathcal{S}$. Specifically, let \(\Vert \cdot \Vert_{\mathbb{R}^s} : \mathbb{R}^s \rightarrow \mathbb{R}\) be defined by \(\Vert v \Vert_{\mathbb{R}^s} = \Vert \theta^{-1}(v) \Vert_{K}\) for \( v \in \mathbb{R}^s \), where $\theta: \mathcal{S} \rightarrow \mathbb{R}^{s}$ is a linear isomorphism. For instance, if $\{v_{1},\dots,v_{s}\}$ is an orthonormal basis of $\mathcal{S}$, $\theta$ can map it to the standard basis $\{e_{1},\dots, e_{s}\}$ in $\mathbb{R}^{s}$. The \( K \)-norm mechanism described in Algorithm \ref{alg: optimal K mech} samples in \(\mathbb{R}^s\) via norm $\Vert \cdot \Vert_{\mathbb{R}^{s}}$. The noise is then transformed back into the ambient space by the isomorphism. 

\begin{algorithm}[htb!]
\caption{\texttt{Optimal $K$-Norm Mechanism}}
\label{alg: optimal K mech}
\begin{algorithmic}[1]
    \STATE \textbf{Input:} Query $\phi$ to be privatized, invariant $t$, privacy parameter $\epsilon$
    \STATE Get $S_{Semi} = \{\phi(X) - \phi(X'): A_{a(t)}(X,X')=1\}$ and set $s = \text{dim}(\Span{(S_{Semi})})$
    \STATE Define $K$-norm in $\mathcal{S} = \Span(S_{Semi})$ where $K = \hull(S_{Semi})$, and isomorphism $\theta: \mathcal{S} \rightarrow \mathbb{R}^{s}$ 
    \STATE Define a norm \(\Vert \cdot \Vert_{\mathbb{R}^s} : \mathbb{R}^s \rightarrow \mathbb{R}\) by \(\Vert v \Vert_{\mathbb{R}^s} = \Vert \theta^{-1}(v) \Vert_{K}\) for \( v \in \mathbb{R}^s \)
    \STATE Sample $V \sim f_{V}(v) \propto \exp \left( -\epsilon \Vert v \Vert_{\mathbb{R}^{s}} \right) $
    \STATE \textbf{Output: } $\phi(X) + \theta^{-1}(V)$
\end{algorithmic}    
\end{algorithm}

\begin{proposition}\label{prop: sensitivity one}
For a query $\phi: \mathcal{D} \rightarrow \mathbb{R}^{d}$, consider the sensitivity space $S_{\text{Semi}} = \{ \phi(X) - \phi(X') : A_{a(t)}(X, X') = 1 \}$. Let $\mathcal{S} = \operatorname{Span}(S_{\text{Semi}})$ be a subspace spanned by $S_{\text{Semi}}$, and let $\theta: \mathcal{S} \rightarrow \mathbb{R}^{s}$ be a linear isomorphism, where $s = \dim(\mathcal{S})$. Define the norm $\Vert \cdot \Vert_{\mathbb{R}^{s}}: \mathbb{R}^{s} \rightarrow \mathbb{R}$ by $ \Vert v \Vert_{\mathbb{R}^{s}} = \Vert \theta^{-1}(v) \Vert_{K}$ for $v \in \mathbb{R}^{s}$, where $\Vert \theta^{-1}(v) \Vert_{K}$ is the $K$-norm with $K = \hull(S_{Semi})$. Then, we have
\[
\Delta_{s} = \sup_{A_{a(t)}(X, X') = 1} \Vert \theta(\phi(X) - \phi(X')) \Vert_{\mathbb{R}^{s}} = 1.
\]
\end{proposition}

\begin{theorem}\label{thm: semi dp K norm mech}(The optimal $K$-norm mechanism for Semi-DP).
The $K$-norm mechanism in Algorithm \ref{alg: optimal K mech} satisfies $(\mathcal{D}_{t},A_{a(t)},f_{\epsilon,0})$-DP and it is the optimal $K$-norm mechanism.
\end{theorem}

The optimality in Theorem \ref{thm: semi dp K norm mech} stems from the fact that $K = \hull(S)$ is the smallest norm ball, in terms of containment order, among all norm balls that include the sensitivity space in $\mathcal{S}$. Then the norm ball induced by $\Vert \cdot \Vert_{\mathbb{R}^s}$ is the smallest norm ball that contains $\theta(S)$ since the isomorphism $\theta$ preserves properties including containment.

Finally, once the sensitivity space is of full-rank with the  dataspace $\mathcal{D}$ and the adjacency function is $A_{1}(X,X')$, our mechanism coincides with the $K$-norm mechanism in \cite{awan2021structure}. Therefore, our approach is indeed a generalization of theirs.

\section{Application to Contingency Table Analysis}\label{sec: contingency table}

A contingency table, or a multi-way frequency table, is a fundamental data summary tool widely used in various fields, such as the DP publications by the US Census Bureau. In this section, we demonstrate how to create a privatized contingency table when true one-way margins are held invariant. We derive an optimized Gaussian mechanism satisfying \(G_\mu\)-semi DP and $K$-norm mechanism satisfying \(f_{\epsilon,0}\)-semi DP, along with numerical verification to show how our approach outperforms naive methods. Finally, we present a Semi-DP UMPU test for the odds ratio in a $2 \times 2$ contingency table.

\subsection{Private Release of Contingency Table under True Margins}\label{sec: contingency theory}

We identify a proper sensitivity space so that mechanisms investigated in Section \ref{sec: mechanism design} can be applied. To begin with, let \(\phi: \mathcal{X}^{n} \rightarrow \mathbb{R}^{r \times c}\) represent an \(r \times c\) contingency table, where the sample space \(\mathcal{X}\) is partitioned into \(r \times c\) categories. The table can be vectorized into a \(d\)-dimensional vector, where \(d = r \times c\). Therefore, the contingency table can be written as \(\phi: \mathcal{X}^{n} \rightarrow \mathbb{R}^{d}\).

The one-way margins of an \(r \times c\) contingency table correspond to the margins for each feature \(A_1\) and \(A_2\), which are described using the notation \(\mathbf{T}(X)\) in Example \ref{sec: example on count statistics}. Specifically, for two features \(A_1\) with \(r\) categories and \(A_2\) with \(c\) categories, the one-way margins are captured by \(\mathbf{T}(X)=(T_{1}(X),T_{2}(X))\), where \(T_1(X)\) and \(T_2(X)\) represent the total counts for each category of \(A_1\) and \(A_2\), respectively.

Recall from Theorem \ref{thm: count stat semi adj para} that the semi-adjacent parameter is bounded by $a(t) \leq 3$ as \(r \times c\) contingency table is counts for $p=2$ features. To express the corresponding sensitivity space $S_{Semi}$, we first introduce some notation.

Define an \(r \times c\) matrix \(\mathbf{v}_{ijkl}\) as 
\[
(\mathbf{v}_{ijkl})_{pq} =
\begin{cases}
1, & \text{if } (p, q) = (i, j) \text{ or } (p, q) = (k,l), \\
-1, & \text{if } (p, q) = (i, l) \text{ or } (p, q) = (k, j), \\
0, & \text{otherwise}.
\end{cases}
\]

In \(\mathbf{v}_{ijkl}\), \(i, k \in \{1, \ldots, r\}\) are the row indices, and \(j, l \in \{1, \ldots, c\}\) are the column indices. The matrix \(\mathbf{v}_{ijkl}\) represents an element of the sensitivity space that adjusts four specific positions within an \(r \times c\) contingency table while preserving the one-way margins. See Figure \ref{fig: n notation} for a typical example of a \(\mathbf{v}_{ijkl}\) matrix. 

\begin{figure}[ht]
    \centering
    \includegraphics[scale=0.47]{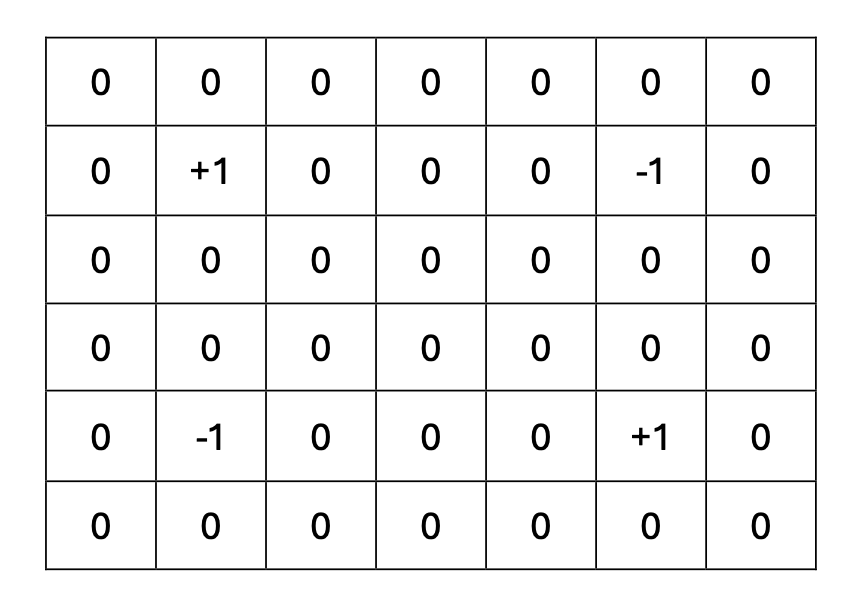}
    \caption{A \( 7 \times 6 \) contingency table illustrating the sensitivity space element \( \mathbf{v}_{ijlk} \) for the indices \( i = 2 \), \( j = 2 \), \( k = 5 \), and \( l = 5 \). The table shows the placement of \( +1 \) values at positions \( (i, j) = (2, 2) \) and \( (l, k) = (5, 5) \), as well as \( -1 \) values at positions \( (i, k) = (2, 5) \) and \( (l, j) = (5, 2) \), while all other cells contain zeros. It represents a typical element of the sensitivity space \( S_{Semi} \) that preserves the one-way margins.}
    \label{fig: n notation}
\end{figure}

\begin{proposition}\label{thm: table sens space}
    Let \(\phi: \mathcal{X}^{n} \rightarrow \mathbb{R}^{r \times c}\) be an \(r \times c\) contingency table and let the invariant be its one-way margins \(t=\mathbf{T}(X)=(T_1(X),T_2(X))\). Then the sensitivity space \(S_{Semi}\) with respect to the adjacency function $A_{3}$ is
    \[
    S_{Semi} = \left\{\mathbf{v}_{ijlk} \in \mathbb{R}^{r \times c} : i, k \in \{1, \ldots, r\}, \; j, l \in \{1, \ldots, c\}, \; i \neq k, \; j \neq l \right\} \cup \{\mathbf{0}\},
    \]
   where $\mathbf{0}$ denotes the zero matrix in $\mathbb{R}^{r \times c}$.
\end{proposition}

The sensitivity space \( S_{Semi} \) in Theorem \ref{thm: table sens space} consists of matrices that preserve the one-way margins. The zero matrix, \(\mathbf{0}\), emerges from scenarios where individual replacements do not affect the contingency table counts. For instance, consider swapping an individual with \( A_1 = 1 \) and \( A_2 = 1 \) for another individual with \( A_1 = 2 \) and \( A_2 = 2 \), followed by switching an individual with \( A_1 = 2 \) and \( A_2 = 2 \) back to an individual with \( A_1 = 1 \) and \( A_2 = 1 \). This process results in no net change to the table. 

The nonzero elements of \( S_{Semi} \) are generated by modifying a total of three individuals, resulting in the placement of two \( 1 \)’s and two \(-1\)’s in positions that preserve both the row and column sums. Each element in \( S_{Semi} \) can be represented as a \( d = r \times c \) dimensional vector through vectorization. Notably, Theorem \ref{thm: table sens space} aligns with Lemma A.1 in \cite{awan2024best}, which identified a basis for the space that preserves all margins for a contingency table.

\begin{example}\label{example: 2 by 2 sens space}
For a \(2 \times 2\) contingency table, where \(i, j, k, l \in \{1, 2\}\), there are only two possible combinations that satisfy \(i \neq k\) and \(j \neq l\): \((i, j, k, l) = (1, 1, 2, 2)\) or \((2, 1, 1, 2)\). Consequently, the nonzero elements in \(S_{Semi}\) are:
1. For \(i = 1, j = 1, k = 2, l = 2\):
   \[
   \mathbf{v}_{1122} = 
   \begin{bmatrix}
   1 & -1 \\
   -1 & 1
   \end{bmatrix} \quad \text{or as a vector: } (1, -1, -1, 1).
   \]

2. For \(i = 2, j = 1, k = 1, l = 2\):
   \[
   \mathbf{v}_{2112} = 
   \begin{bmatrix}
   -1 & 1 \\
   1 & -1
   \end{bmatrix} \quad \text{or as a vector: } (-1, 1, 1, -1).
   \]

This matches what we observed in the \(2 \times 2\) contingency table in Section \ref{sec: example on count statistics}.
\end{example}

Moreover, the sensitivities are given as \(\Delta_{1}(S_{Semi}) = 4\), \(\Delta_{2}(S_{Semi}) = 2\), and \(\Delta_{\infty}(S_{Semi}) = 1\), regardless of \( r \) and \( c \), since each nonzero vector in \( S_{Semi} \) contains two \(1\)s and two \(-1\)s, with zeros in all other positions. 

Having established the sensitivity space \( S_{Semi} \), we can now apply the mechanisms from Section \ref{sec: mechanism design}. For Gaussian mechanism in Algorithm \ref{alg: gaussian mechanism}, we leverage the sensitivity space \( S_{Semi} \) in Theorem \ref{thm: table sens space} and the \(\ell_{2}\) sensitivity \(\Delta_{2}(S_{Semi}) = 2\). 

\begin{example}
    Gaussian mechanism in Algorithm \ref{alg: gaussian mechanism} for $2 \times 2$ contingency table is deployed with $\Delta_{2}(S_{Semi})=2$ and 
    \begin{equation*}
    P = \frac{1}{4} \begin{pmatrix} 1 & -1 & -1 & 1 \\ -1 & 1 & 1 & -1 \\ -1 & 1 & 1 & -1 \\ 1 & -1 & -1 & 1 \end{pmatrix}.
\end{equation*}
This is because for $s = \text{dim}(\text{span}(S_{Semi}))$, if $\{u_{1},\cdots u_{s}\}$ is the basis of $\text{span}(S_{Semi})$, the projection matrix $P$ can be constructed by $P = \sum_{i}^{s}u_{i}u_{i}^\top$ and the basis in this case is $\{(1,-1,-1,1)\}$.
\end{example}

Likewise, we can deploy Algorithm \ref{alg: optimal K mech} with \( K = \hull(S_{Semi}) \). However, implementing the $K$-norm mechanism is not straightforward. The rank-deficiency in sensitivity space raises a technical difficulty as the corresponding convex hull is degenerate in the ambient space. To deal with this, we take a cube of the same dimension of $S_{Semi}$ that contains $K$ and sample from the cube to implement a rejection sampler. Recall from Proposition \ref{prop: sensitivity one} that $\Delta_{s}(S_{Semi})=1$.

\begin{algorithm}[htb!]
\caption{\texttt{Optimal $K$-Norm Mechanism with Rejection Sampling}}\label{alg: knorm mech reject sampling}
\begin{algorithmic}[1]
    \STATE \textbf{Input:} Query $\phi$, sensitivity space $S_{Semi}$, sensitivity $\Delta_{K}(S_{Semi})$ and privacy parameter $\epsilon$
    \STATE Find basis $\{b_{i}\}_{i=1}^{s}$ for $\Span(\hull(S_{Semi}))$
    \STATE Sample $r \sim \text{Gamma}(\alpha = s+1,\beta=\epsilon/\Delta_{K}(S_{Semi})$
    \STATE Find scaling constants $c_{1},\cdots,c_{k}$ such that a cube $\prod_{i=1}^{s}[-c_{i},c_{i}]$ contains $\hull(S_{Semi})$ in the coordinates $\{b_{i}\}_{i=1}^{s}$
    \STATE Sample $U_{j} \stackrel{iid}{\sim} U(-1,1)$ for $j=1,\cdots,s$
    \STATE Set $V = \sum_{i=1}^{s}U_{i}c_{i}b_{i}$
    \STATE If $V \in \hull(S)$, output $\phi(X) + rV$, else go to 5
    
\end{algorithmic}    
\end{algorithm}

The $K$-norm mechanism in Algorithm \ref{alg: knorm mech reject sampling} works by finding a basis for the span of the convex hull of the sensitivity space \( \hull(S_{Semi}) \) and generating noise from a cube enclosing the convex hull of the sensitivity space. Rejection sampling ensures that the generated noise lies within the convex hull, and thus satisfies the privacy constraints.

\begin{remark}
Note that one may apply additional transformations to the mechanism's output through various post-processing methods. For example, while preserving the margins, each count can be adjusted to a nonnegative integer value. Since these transformations use only the private output and publicly available invariant, without accessing the raw data, the same privacy guarantee is still satisfied by the post-processing property in Theorem \ref{thm: semi dp postprocessing}.    
\end{remark}
    
\subsection{Numerical Experiments for additive noise}\label{sec: contingency numerical}
In this section, we numerically compare the \( L_{2} \)-costs of our mechanisms with those of naive designs. We treat the contingency table \( \phi(X) \) as fixed, with its entries sampled from a multinomial distribution, \(\text{Multinomial}_{k^{2}}(n, \pi)\). For our comparisons, we set \( n = 500 \) and consider two models for the true probabilities \(\pi = (\pi_{1}, \ldots, \pi_{k^{2}})\). In Model I, we assume a uniform probability for each cell, setting \(\pi_{i} = 1/k^{2}\) for all \(i\). In Model II, the cell probabilities increase linearly, defined by \(\pi_{i} = i / \left(\sum_{i=1}^{k^{2}} i\right)\). To be consistent with the approach in Section \ref{sec: contingency theory}, we assume that all one-way margins are published. Here, we calculate $a(t)=3$ to consider the privacy guarantee for the worst-case. We report the average \( L_{2} \)-loss \(\Vert M(X) - \phi(X) \Vert_{2}\) aggregated over 30 replicates to assess the performance of the mechanisms. \\

\noindent
{\bf Comparison on Gaussian Mechanisms:}
While running our Gaussian mechanism from Algorithm \ref{alg: gaussian mechanism}, we compare it to a naive Gaussian mechanism defined as:
\[
M_{\text{Naive}}(X) = \phi(X) + N(0, (3\sqrt{2}/\mu)^2 I).
\]
This naive mechanism uses an \(\ell_2\) sensitivity of \(\sqrt{2}\), as previously observed, and a semi-adjacency parameter of \(a(t) = 3\). Since any mechanism that satisfies \((\mathcal{D}, A_{1}, G_{\mu/3})\)-DP also satisfies \((\mathcal{D}, A_{3}, G_{\mu})\)-DP by the group privacy property, and further satisfies \((\mathcal{D}_t, A_{3}, G_{\mu})\)-DP, we adjust the Gaussian mechanism from Proposition \ref{prop: standard gaussian mechanism} by substituting \(\mu/3\) in place of \(\mu\). Therefore, this naive approach effectively inflates the sensitivity by the semi-adjacency parameter when applying the Gaussian noise.

\begin{figure}[htbp]
    \centering
    \includegraphics[scale=0.47]{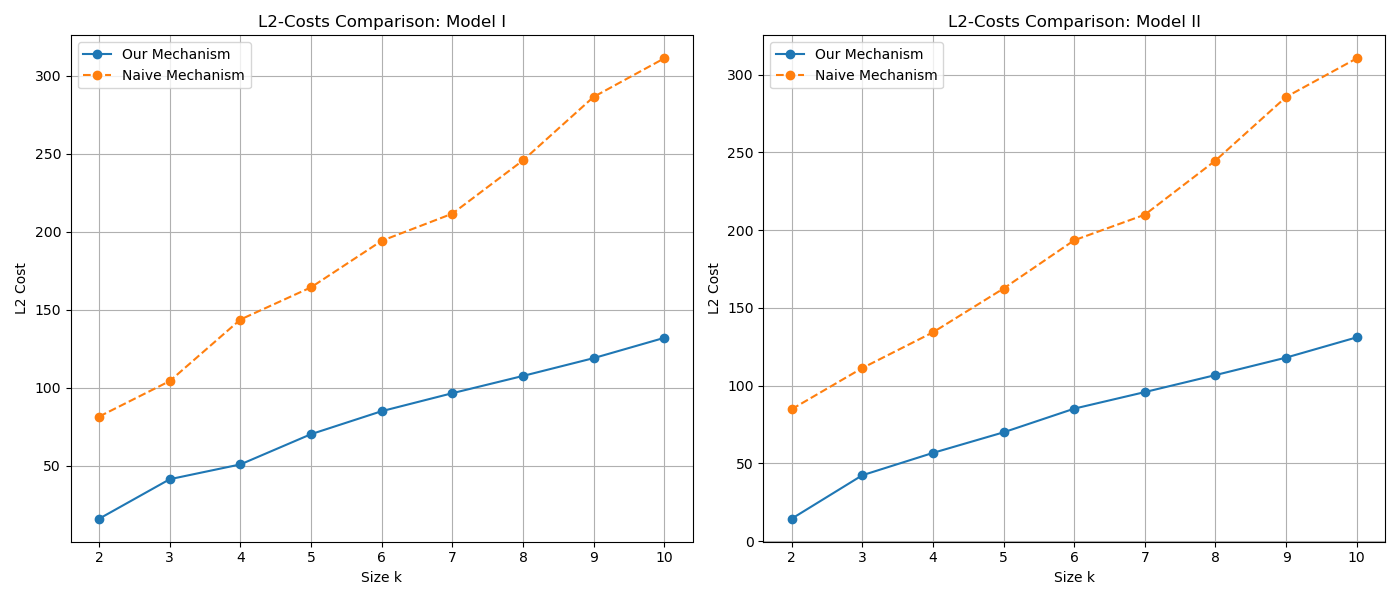}
    \caption{Comparison of average \( L_2 \)-costs between our mechanism and the naive mechanism across varying contingency table sizes \( k \in \{2,3,\dots,10\} \). Model I assumes uniform cell probabilities, while Model II has linearly increasing cell probabilities.}
    \label{fig: simul gaussian}
\end{figure}

Figure \ref{fig: simul gaussian} demonstrates that our Gaussian mechanism from Algorithm \ref{alg: gaussian mechanism} consistently results in smaller \( L_2 \) costs compared to the naive approach across both models. This aligns with theoretical expectations, as our mechanism applies Gaussian noise with a smaller covariance. Specifically, the covariance of the Gaussian noise in our mechanism, \((2/\mu)^{2}P_{k \times k}\), is smaller than that used in the naive approach, \((3\sqrt{2}/\mu)^{2}I_{k \times k}\). This is due to the fact that \( 2/\mu < 3\sqrt{2}/\mu \) and \( P_{k \times k} \preceq I_{k \times k} \), where \( P_{k \times k} \) is the projection matrix.
\\

\noindent
{\bf Comparison on $K$-Norm Mechanisms:} To numerically compare the $K$-norm mechanisms, we consider different contingency table sizes by taking $k \in \{2,3\}$ with different privacy parameters $\epsilon \in \{0.1, 0.5, 1\}$.

\begin{figure}[h]
    \centering
    \includegraphics[scale=0.45]{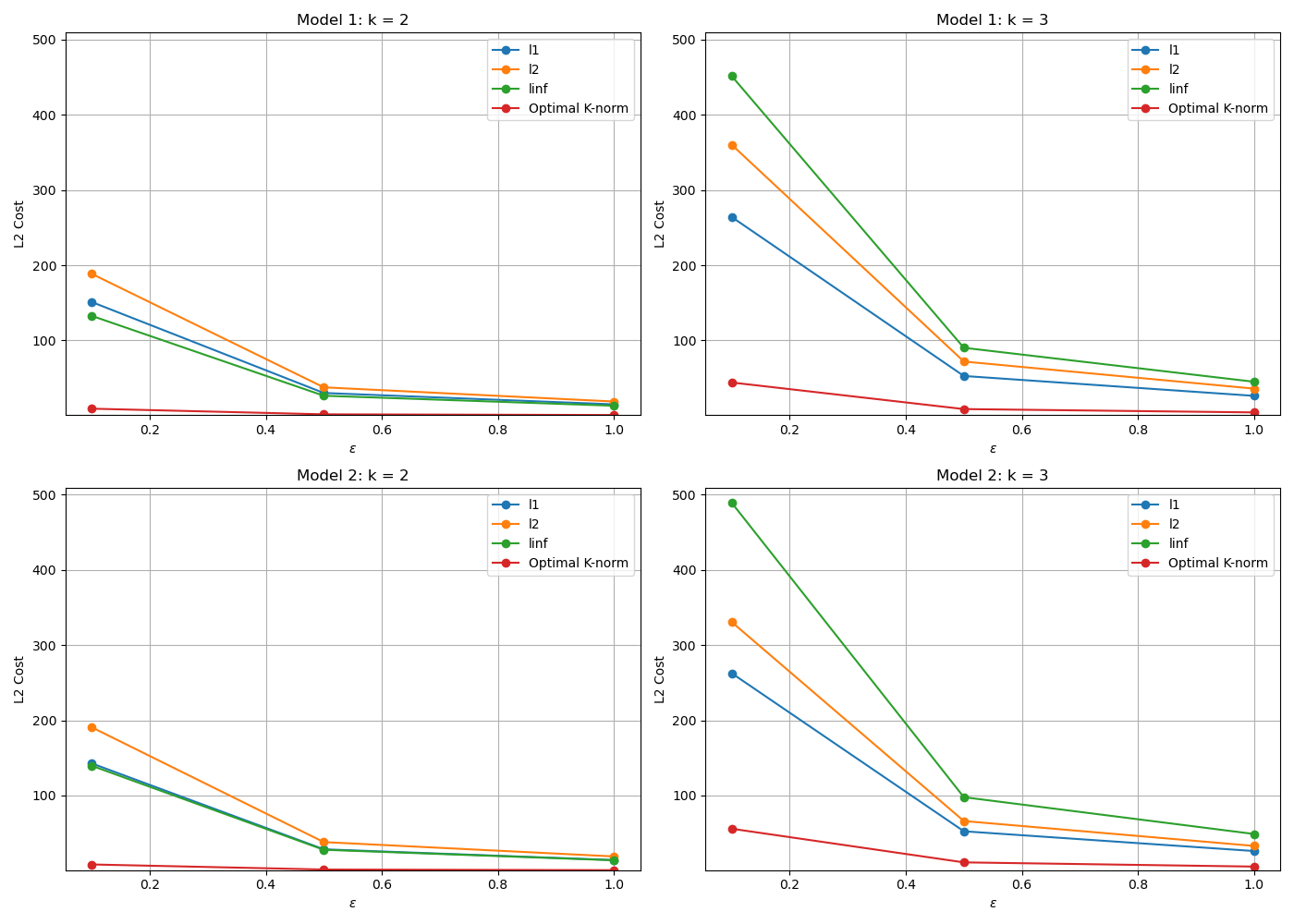}
    \caption{Comparison of \( L_2 \)-costs for the \( K \)-norm mechanism against naive \(\ell_{1}\), \(\ell_{2}\), and \(\ell_{\infty}\)-norm mechanisms across two models for contingency tables of size \( k = 2 \) and \( k = 3 \). The results are shown for three privacy parameters \(\epsilon = 0.1, 0.5,\) and \( 1\).}
    \label{fig: simul knorm}
\end{figure}

We implement the optimal $K$-norm mechanism by Algorithm \ref{alg: knorm mech reject sampling} with the sensitivity space $S_{Semi}$ as identified in Theorem \ref{thm: table sens space}. 
On the other hand, we consider \(\ell_{1}, \ell_{2}, \ell_{\infty}\)-norm mechanisms in a naive way. We run $\ell_{1},\ell_{2}$ and $\ell_{\infty}$-norm mechanisms in Section \ref{app: details on simulation} in the Appendix with the sensitivity space $S_{DP} = \{\phi(X)-\phi(X'): A_{1}(X,X')=1\}$. Note that \(\ell_{1}, \ell_{2}\), and \(\ell_{\infty}\)-sensitivities from $S_{DP}$ are \(2\), \(\sqrt{2}\), and \(1\), respectively, since we can easily check that each element of $S_{DP}$ is a vector with a single $1$ and a single $-1$. On the other hand, we take the privacy parameter $\epsilon/3$ so that the naive mechanisms satisfy \((\mathcal{D}, A_{1}, f_{\epsilon/3,0})\). By doing so, they satisfy \((\mathcal{D}, A_{3}, f_{\epsilon/,0})\)-DP, and immediately  \((\mathcal{D}_{t}, A_{3}, f_{\epsilon/,0})\)-DP. Same as what we have observed in the naive Gaussian mechanism, the naive mechanisms via group privacy use sensitivities inflated by a factor of 3 to achieve group privacy with \(a(t) = 3\).

Finally, Figure \ref{fig: simul knorm} shows see the efficiency of the optimal $K$-norm mechanism, which achieves significantly lower $L_{2}$-cost compared to other naive mechanisms. 

\subsection{Semi Private UMPU Testing}

This section formalizes the construction of the UMPU test for odds ratio testing under Semi-DP given the one-way margins as the invariant. We extend the framework of semi-private hypothesis testing introduced in \cite{awan2023canonical}, where the concept was employed to establish upper bounds on the power of private hypothesis testing. Specifically, \cite{awan2023canonical} demonstrated that while no UMPU test satisfying DP exists in general, it is achievable under Semi-DP. Their results were focused on tests for the difference of proportions in Bernoulli data. We extend this approach to testing the odds ratio in a $2 \times 2$ contingency table with the invariant as the one-way margins. By refining the definition of semi-privacy, we also provide a more simplified proof.

Denote a $2 \times 2$ contingency table by $(x_{11},x_{12},x_{21},x_{22})$ and let the one-way margins $t =(t_{1\cdot},t_{2 \cdot},t_{\cdot 1},t_{\cdot 2})^\top$ be the invariant, where $t_{i\cdot}=\sum_{j=1}^{2} x_{ij}$ and $t_{\cdot, j} = \sum_{i=1}^{2} x_{ij}$. We test the odds ratio $w = \frac{\theta_{2}/1-\theta_{2}}{\theta_{1}/1-\theta_{1}}$:

\begin{center}
$H_0$: $w \leq 1$\quad vs. \quad $H_1$: $w>1$,
\end{center}
while preserving the privacy of individuals in the database.

We follow the framework of $f$-DP hypothesis testing as investigated in \cite{awan2023canonical}. One can define a test to be a function $\phi:\mathcal{D} \rightarrow [0,1]$, where $\phi(X)$ represents the probability of rejecting the null hypothesis given the database $X$. According to \cite{awan2023canonical}, the mechanism corresponding to this test releases a random variable drawn as $\text{Bern}(\phi(X))$, where $1$ represents ``Reject'' and $0$ represents ``Accepts''. The test $\phi$ satisfies $(\mathcal{D},A,f)$-DP if the corresponding mechanism $\text{Bern}(\phi(X))$ satisfies $(\mathcal{D},A,f)$-DP. 

\begin{lemma}\label{thm: f-dp test}(Lemma 4.1 in \cite{awan2023canonical}). Let $\mathcal{D}$ be a dataspace, $A: \mathcal{D} \times \mathcal{D}$ be an adjacency function and $f$ be a tradeoff function. A test $\phi: \mathcal{D} \rightarrow [0,1]$ satisfies $(\mathcal{D},A,f)$-DP if and only if $\phi(X) \leq 1-f(1-\phi(X'))$ for all $X,X' \in \mathcal{D}$ such that $A(X,X') = 1$.
\end{lemma}

Lemma \ref{thm: f-dp test} shows that a test function satisfies $(\mathcal{D},A,f)$-DP if for any databases $X,X' \in \mathcal{D}$ such that $A(X,X')=1$, the values $\phi(X)$ and $\phi(X')$ are close in terms of an inequality based on $f$.

Recall from Example \ref{example: 2 by 2 semi adj para} that when the  one-way margins are invariant, $a(t) \leq 3$. From Example \ref{example: 2 by 2 sens space}, we know that the sensitivity space for a query that publishes the $2 \times 2$ contingency table given this invariant is 
\[
S_{Semi} = \{(1, -1, -1, 1),(-1, 1, 1, -1),(0,0,0.0)\}.
\]
Thus, by Lemma \ref{thm: f-dp test}$, \phi$ is a $(\mathcal{D}_{t},A_{a(t)},f)$-DP test if
\begin{equation}\label{equ: f-semi dp condition}
\begin{split}
    & \phi(x_{11},x_{12},x_{21},x_{22}) \leq 1-f(1-\phi(x_{11}+1,x_{12}-1,x_{21}-1,x_{22}+1)),\\
    & \phi(x_{11},x_{12},x_{21},x_{22}) \leq 1-f(1-\phi(x_{11}-1,x_{12}+1,x_{21}+1,x_{22}-1)).
\end{split}    
\end{equation}

In \cite{awan2023canonical}, they utilized a canonical noise distribution (CND) for $f$-DP to construct optimal tests. A CND is formulated to satisfy $f$-DP for any given tradeoff function $f$ by optimally matching the tradeoff function. After being introduced by \cite{awan2023canonical}, CNDs have been further investigated in \cite{awan2022log} and \cite{awan2023optimizing}.

\begin{definition}\label{def: cnd}(Canonical noise distribution (CND) \citep{awan2023canonical}). Let $f$ be a symmetric nontrivial tradeoff function. A continuous distribution function $F$ is a canonical noise distribution(CND) for $f$ if 
\begin{itemize}
    \item[1. ] for every statistic $\phi: \mathcal{D} \rightarrow \mathbb{R}$ with sensitivity $\Delta > 0$, and $N \sim F(\cdot),$ the mechanism $\phi(X) + \Delta N$ satisfies $f$-DP. Equivalently, for every $m \in [0,1],$ $T(F(\cdot),F(\cdot-m)) \geq f$,
    \item[2. ] $f(\alpha)=T(F(\cdot),F(\cdot-1))(\alpha)$ for all $\alpha \in (0,1)$,
    \item[3. ] $T(F(\cdot),F(\cdot-1))(\alpha) = F(F^{-1}(\alpha)-1)$ for all $\alpha \in (0,1)$,
    \item[4. ] $F(x) = 1-F(-x)$ for all $x \in \mathbb{R};$ that is, $F$ is the cdf of a random variable which is symmetric about zero.
\end{itemize}
\end{definition}

In Definition \ref{def: cnd}, the four properties can be interpreted as follows: 1) adding noise scaled by the sensitivity \(\Delta\) to a statistic with sensitivity \(\Delta\) ensures that the resulting mechanism satisfies \((\mathcal{D},A,f)\)-DP; 2) when the statistics for two datasets differ exactly by the sensitivity \(\Delta\), the tradeoff function \(T\) between the two noisy statistics equals \(f\), indicating that the privacy protection is optimally tight; 3) for such statistics as in part 2, the optimal rejection region follows a threshold form, which demonstrates a monotone likelihood ratio property; and 4) because the \(f\)-DP guarantee is symmetric, we restrict our attention to symmetric distributions. For example, $N(,1/\mu^{2})$ is a CND for $(\mathcal{D},A,G_{\mu})$-DP and Tulap distribution is the unique CND for $(\mathcal{D},A,f_{(\epsilon,\delta))}$-DP \citep{awan2022log}.

\cite{awan2023canonical} proved that a CND exists for any symmetric nontrivial tradeoff function $f$ and gave the following construction, which can be easily sampled by inverse transform sampling:

\begin{lemma}(CND Construction: Theorem 3.9 in \citep{awan2023canonical}). Let $f$ be a symmetric nontrivial tradeoff function. Let $c \in [0,1/2)$ be the unique value satisfying $f(1-c)=c$. Define $F_{f}: \mathbb{R} \rightarrow \mathbb{R}$ as 
\[
F_f(x) = 
\begin{cases} 
    f(F_f(x + 1)) & \text{if } x < -1/2, \\
    c(1/2 - x) + (1 - c)(x + 1/2) & \text{if } -1/2 \leq x \leq 1/2, \\
    1 - f(1-F_f(x - 1)) & \text{if } x > 1/2.
\end{cases}
\]
Then $F_{f}$ is a canonical noise distribution for $f$.  
\end{lemma}

\cite{awan2023canonical} showed that CNDs are central to the construction of optimal DP tests.

Recall that our objective is to develop a private test that satisfies Semi-DP while maximizing statistical power. Without the privacy constraint, there is no UMP test for this problem. Traditionally, attention is restricted to unbiased tests. Recall that for a parameter space \(\Omega\) that can be partitioned into \(\Omega = \Omega_{0} \cup \Omega_{1}\), a test is unbiased if for all \( w_{1} \in \Omega_{1} \) and \( w_{0} \in \Omega_{0} \), the power at \( \Omega_{1} \) is higher than at \( w_{0} \). Then the search for a UMP unbiased test can be restricted to tests which satisfy \( \mathbb{E}_{w=1}\left( \phi(X_{11}, X_{12}, X_{21}, X_{22}) \mid T \right) = \alpha \) by leveraging the Neyman structure as in \cite{awan2023canonical}.

Note that \( (x_{11}, x_{12}, x_{21}, x_{22}) \mid t \) is equal in distribution to \( \left( H, \; t_{1\cdot} - H, \; t_{\cdot 1} - H, \; H + t_{2\cdot} - t_{\cdot 1} \right) \), where \( H \sim \text{Hyper}(t_{1\cdot}, t_{2\cdot}, t_{\cdot 1}, w) \) is the Fisher noncentral hypergeometric distribution, which has the probability mass function (pmf):

\begin{equation*}
    \mathbb{P}_{w}\left( H = x \right) = \frac{\binom{t_{1\cdot}}{x} \binom{t_{2\cdot}}{t_{\cdot 1} - x} w^{x}}{\sum_{k=L}^{U} \binom{t_{1\cdot}}{k} \binom{t_{2\cdot}}{t_{\cdot 1} - k} w^{k}},
\end{equation*}
with support \( L = \max\{ 0, \; t_{\cdot 1} - t_{2\cdot} \}, \; L + 1, \ldots, \; U = \min\{ t_{1\cdot}, \; t_{\cdot 1} \} \). Additionally, \( x_{11} \mid t_{1\cdot} \sim \text{Binom}(t_{1\cdot}, \theta_{1}) \) and \( x_{21} \mid t_{2\cdot} \sim \text{Binom}(t_{2\cdot}, \theta_{2}) \).

This leads to tests \( \phi \) satisfying:
\begin{equation}\label{equ: neyman umpu size alpha}
    \mathbb{E}_{H \sim \text{Hyper}(t_{1\cdot}, t_{2\cdot}, t_{\cdot 1}, 1)} \phi\left( H, \; t_{1\cdot} - H, \; t_{\cdot 1} - H, \; H + t_{2\cdot} - t_{\cdot 1} \right) = \alpha,
\end{equation}
as shown in Lemma \ref{thm: neyman umpu size alpha}.

Then the Semi-DP UMPU test is equivalent to identifying the UMP among the following set of \( (\mathcal{D}_{t}, A_{a(t)}, f) \)-DP tests:
\begin{equation}
    \Phi_{f}^{\text{semi}} = \left\{ \phi(x_{11}, x_{12}, x_{21}, x_{22}): \phi \text{ satisfies Equation } (\ref{equ: f-semi dp condition}) \text{ and Equation } (\ref{equ: neyman umpu size alpha}) \right\}.
\end{equation}

\begin{theorem}\label{thm: semi dp umpu}
Let \( f \) be a symmetric nontrivial tradeoff function, and let \( F_{f} \) be a CND for \( f \). Let \( \alpha \in (0,1) \) be given. For the \( 2 \times 2 \) contingency table and the hypothesis \( H_{0}: \; w \leq 1 \) vs. \( H_{1}: \; w > 1 \),

\begin{equation}
    \phi^{\ast}(X) = \psi_{z}^{\ast}(x_{11}) = F_{f}\left( x_{11} - m(t) \right),
\end{equation}
is the UMP among \( \Phi_{f}^{\text{semi}} \), where \( m(t) \) is chosen such that the test has size \( \alpha \). Moreover, if \( N \sim F_{f} \), then  \( U = X_{11} + N \) satisfies \( (\mathcal{D}_{t}, A_{a(t)}, f) \)-DP, and the quantity
\begin{equation}
    p = \mathbb{E}_{X_{11}\sim \text{Hyper}(t_{1\cdot}, t_{2\cdot}, t_{\cdot 1}, 1)}F_{f}(X_{11}-U),
\end{equation}
which is a post-processing of $U$, is a $p$-value that agrees with $\phi^{\ast}(X)$. That is, $\mathbb{P}_{U}\left( p \leq \alpha \mid X \right) = \phi^{\ast}(X).$
\end{theorem}

In Theorem \ref{thm: semi dp umpu}, only \( x_{11} \) is actually utilized among the four variables. However, it does not necessarily have to be \( x_{11} \); any of the variables \( x_{11}, \ldots, x_{22} \) can be free. This is because there are three independent constraints imposed by the given margins \( t \), resulting in one degree of freedom. Once a value, such as \( x_{11} \), is determined, the remaining three counts are subsequently determined by \( x_{11} \) and the margins \( t \).

We remark that both our derivation and the original semi DP test in \cite{awan2023canonical} leverage a DP version of the classic Neyman-Pearson Lemma, originally developed in \cite{awan2018differentially}. However, unlike \cite{awan2023canonical} which leverages group privacy type reasoning on their construction, we directly target the UMP test among the tests with Neyman structure under the Semi-DP constraint. The simplified proof highlights the value of understanding the Semi-DP framework.

\section{Analysis on US 2020 Census Redistricting Data (P.L. 94-171)}\label{sec: census}

The US Census published noisy counts on demographics while giving some true counts in the 2020 Census Redistricting Data (P.L. 94-171) Summary File.\footnote{\url{https://www.census.gov/programs-surveys/decennial-census/about/rdo/summary-files.html\#P1}} In the 2020 Census Disclosure Avoidance System (DAS) and the TopDown Algorithm (TDA), the invariants consist of the following counts:
\begin{itemize}
    \item total population for each state, the District of Columbia, and Puerto Rico,
    \item the number of housing units in each block (but not the population living in theses housing units),
    and
    \item the count and type of occupied group quarters in each block (but not the population living in theses group quarters).
\end{itemize}

These invariants are crucial for specific operations and legal requirements, such as the apportionment of the House of Representatives and maintaining the Master Address File (MAF) used in the census. 

Rather than $f$-DP or $(\epsilon,\delta)$-DP, the baseline DP definition in the 2020 Decennial Census is $\rho$-zero Concentrated DP($\rho$-zCDP) that utilizes R\'enyi divergence to measure the similarity on outputs of mechanisms. $\rho$-zCDP is slightly weaker than $\mu$-GDP as $\mu$-GDP implies $\frac{\mu^{2}}{2}$-zCDP, but not the opposite \citep{dong2022gaussian}. \citet{su20242020} recently derived an optimized privacy accounting for the US Census products using $f$-DP, but we will restrict our analysis to $\rho$-zCDP for simplicity.

\begin{definition}\label{def: zcdp}(zero-Concentrated Differential Privacy (zCDP), \citep{bun2016concentrated}) Given $\rho \geq 0$, a randomized mechanism $M: \mathcal{D} \rightarrow \mathcal{Y}$ satisfies $(\mathcal{D},A,\rho)$-zero-concentrated differential privacy($(\mathcal{D},A,\rho)$-zCDP) if for all $X,X' \in \mathcal{D}$ such that $A(X,X')=1$, and all $\alpha \in (1,\infty):$
\begin{equation*}
    D_{\alpha}(M(X) \Vert M(X')) \leq \rho \alpha,
\end{equation*}
where $D_{\alpha}(P \Vert Q) = \frac{1}{\alpha -1}\log \left(\sum_{E \in \mathcal{Y}}P(E)^{\alpha}Q(E)^{(1-\alpha)}\right)$ is the R\'enyi divergence of order $\alpha$ of the distribution $P$ from the distribution $Q$.
\end{definition}

We can adopt our privacy framework to the zCDP setting as well. Given the invariant $t$ and its corresponding semi-adjacent parameter $a(t)$, we write a mechanism $M$ satisfies $(\mathcal{D}_{t},A_{a(t)},\rho)$-zCDP when for all $X,X' \in \mathcal{D}_{t}$ such that $A_{a(t)}=1$, $D_{\alpha}(M(X) \Vert M(X')) \leq \rho \alpha$ for all $\alpha \in (1,\infty)$.

To facilitate the comparison of privacy guarantees in different frameworks, we utilize the following two lemmas. Lemma \ref{thm: group privacy in zcdp} establishes group privacy in zCDP while Lemma \ref{thm: zcdp conversion} provides the conversion from $\rho$-zCDP guarantantee to $(\epsilon,\delta)$-DP guarantee. Together, these lemmas offer the necessary tools for our analysis in the following section.

\begin{lemma}\label{thm: group privacy in zcdp}(Group privacy in zCDP, \citep{bun2016concentrated}) Let $M:\mathcal{D} \rightarrow \mathcal{Y}$ satisfy $\rho$-zCDP. Then $M$ guarantees $(k^{2}\rho)$-zCDP for groups of size $k$, that is, for every $X,X' \in \mathcal{D}$ differing in up to $k$ entries and every $\alpha \in (1,\infty)$, we have 
\begin{equation*}
    D_{\alpha}(M(X) \Vert M(X')) \leq (k^{2}\rho) \alpha.
\end{equation*}
\end{lemma}

\begin{lemma}\label{thm: zcdp conversion}(\citep{bun2016concentrated})
If a randomized algorithm $M$ satisfies $\rho-$zCDP, it satisfies ($\epsilon = \rho + 2\sqrt{\rho \log(1/\delta)},\delta)-$DP for all $\delta >0.$ 
\end{lemma}

According to \cite{abowd20222020}, for the production run of the 2020 Census Redistricting Data (P.L. 94-171) Summary File, a total privacy-loss budget of $\rho = 2.63$ was used with $\rho = 2.56$ used for person tables and $\rho = 0.07$ for housing-units tables. The exact allocation of \(\rho\) for a specific query at each geographical level can be calculated using the total allocation of \(\rho\) and the proportions provided in \cite{abowd20222020}.

\subsection{Semi-DP Analysis of the US Census Product}
In this section, we analyze the privacy guarantees of the US Census product using the Semi-DP framework. Specifically, we focus on the total population counts for each state as an invariant and demonstrate how the Census's mechanism provides privacy guarantees based on semi-adjacency accounting for the presence of the invariant. Note that the other two counts---of housing units and group quarters---are not about individual citizens, which is the unit of privacy protection. We also compare the privacy guarantees in both the zCDP and \((\epsilon, \delta)\)-DP frameworks.

We can analyze the privacy guarantee of the US Census product using the semi-adjacent parameter \( a(t) \). Since the total population count for each state is an example of a one-way margin involving a single feature, we have \( a(t) = 2 \), as established by Theorem \ref{thm: count stat semi adj para}. 

With \( a(t) = 2 \), the Census mechanism satisfies \((\mathcal{D}_{t}, A_{2}, 10.24)\)-zCDP, rather than satisfying \((\mathcal{D}, A_{1}, 2.56)\)-zCDP as the Census announced. The parameter 10.24 arises from the group privacy property of \(\rho\)-zCDP, where the privacy parameter is inflated by a factor of \(2^{2} \rho\) for a group of size 2, as described in Lemma \ref{thm: group privacy in zcdp}. Furthermore, satisfying group privacy of size 2 under \(\mathcal{D}\) and \( A_{2} \) implies Semi-DP under \(\mathcal{D}_{t}\) and \( A_{2} \).

Additionally, when converting the \(\rho\)-zCDP guarantee to an \((\epsilon, \delta)\)-DP guarantee with \(\delta=10^{-10}\), the guarantee advertized by Census is \(\epsilon = 17.91528\), while our approach yields \(\epsilon = 40.95057\), both derived using Lemma \ref{thm: zcdp conversion}. 

According to our Semi-DP framework, the Census is  paying a higher privacy cost than what they have explained as the advertised \((\mathcal{D}, A_{1}, 2.56)\)-zCDP guarantee does not account for the impact of the invariant.

\section{Conclusion and Discussion}\label{sec: conclusion}
This work introduces a general Semi-DP framework, building on the concept introduced in \cite{awan2023canonical}, as an extension of traditional DP to address scenarios where true statistics, or invariant, are released alongside DP outputs. Semi-DP redefines adjacency relations to focus on invariant-conforming datasets, providing privacy protection in complex situations that existing DP approaches do not fully address. However, Semi-DP has limitations, particularly in maintaining individual-level privacy, which is a key strength of traditional DP.

The main limitation of Semi-DP lies in the restricted notion of adjacency. In traditional DP, using the adjacency function \( A_{1} \), any individual \( x_{i} \) in the confidential dataset can be replaced with any other individual \( y \), thereby guaranteeing individual-level protection. In contrast, under Semi-DP, the choices for \( y \) are limited to the set \( \mathcal{D}_{t}^{i} \), where the candidates are constrained to individuals who belong to datasets that conform to the same invariant. This restriction reduces the flexibility in substituting \( x \), potentially weakening privacy protection.

Specifically, if an adversary possesses side information indicating that certain candidates \( y \) are not actually present in the dataset, a privacy breach could occur when \( x \) attempts to ``pretend" to be one of these excluded candidates. In such cases, the adversary would immediately recognize the substitution as a false representation, thus compromising the privacy of \( x \). This scenario illustrates how Semi-DP may fail to provide the same level of privacy protection as traditional DP, where individual-level guarantees remain robust even if an adversary knows \( n-1 \) entries in the dataset.

This discussion extends to the implications of side information that an adversary may possess. In standard DP, side information generally does not compromise privacy guarantees because the protection is inherently robust against disclosure of individual entries. However, in Semi-DP, the presence of side information, especially when combined with an invariant, can pose a significant threat to privacy. In the extreme, there can be some side information that, where merged with the invariant, narrows down the databases to a single database, thus completely undermining privacy. 

Overall, while Semi-DP provides a novel approach to handling complex privacy scenarios involving invariants, it requires careful consideration of its limitations and the potential risks associated with side information. Future work may develop strategies to quantify and mitigate these vulnerabilities, such as refining the definition of adjacency or gaining a deeper understanding of the structure of a given invariant to better characterize the trade-offs between privacy guarantees and the release of sensitive information.

\acks{This research was supported in part by the NSF grant SES-2150615. We would like to express our gratitude to Sewhan Kim and Jinwon Sohn for their contributions during the early stages of this project. }


\appendix

\section{Proofs}\label{sec: appendix proof}
\subsection{Proof of Theorem \ref{thm: count stat semi adj para}}

\begin{proof}
Let \( X^* \) be the confidential database consisting of \( n \) individuals, and let the one-way margins \( t = \mathbf{T}(X^*) \) be the invariant. The semi-adjacent parameter \( a(t) \) measures the worst-case number of changes required to change an individual \( x \) into another individual \( y \), while ensuring that the invariant \( t \) are preserved.

Consider any individual \( x \in X^* \). To address the worst-case scenario, choose \( y \) such that all \( p \) features of \( y \) differ from those of \( x \), meaning \( y_{i} \neq x_{i} \) for each feature \( i = 1, \dots, p \). Changing \( x \) to \( y \) involves changing each feature of \( x \), denoted \( x_{i} \), to the corresponding \( y_{i} \).

Note on the other hand that since \( y \) is chosen from the set of possible data values across datasets, there exist individuals \( x^{(1)}, \dots, x^{(p)} \in X^* \) such that, for each feature \( i \), \( x_{i}^{(i)} = y_{i} \). This ensures that there is at least one individual in the dataset who shares feature \( i \) with \( y \).

Therefore, we can proceed the following to preserve the one-way margins: For each feature \( i \), choose \( x^{(i)} \in X^* \) and change this individual to someone who has \( x_{i} \) for feature \( i \), while keeping all other features the same as those of \( x^{(i)} \). This modification ensures that the value \( y_{i} \) is replaced with the original \( x_{i} \), while leaving the counts for other features unchanged, thereby maintaining the one-way margins. This adjustment requires \( p \) additional changes, one for each feature.

Therefore, the total number of changes is \( p + 1 \), leading to the upper bound:
\[
a(t) \leq p + 1.
\]
\end{proof}

\subsection{Proof of Theorem \ref{thm: semi dp composition}}

The proof of Theorem \ref{thm: semi dp composition} relies primarily on Lemma \ref{lemma: tradeoff function}, from \cite{awan2022log}.

\begin{lemma}\label{lemma: tradeoff function}(Proposition A.7 in \cite{awan2022log}.)
    Let $f$ and $g$ be tradeoff functions. Then 
    \begin{equation*}
        (f \otimes g)^{\circ k} \leq f^{\circ k} \otimes g^{\circ k}.
    \end{equation*}
\end{lemma}

By induction, Lemma \ref{lemma: tradeoff function} can be extended to any finite number of tradeoff functions.\\

\begin{proof} 
First of all, it is clear that both ways of analysis share the same indistinguishable pairs:
\begin{center}
$IND(\mathcal{D}_{t},A_a) = \{(X,X') \in \mathcal{D}_{t} \times \mathcal{D}_{t}: A_{a}(X,X')=1\}$,    
\end{center}
 as $IND$ depends only on the invariant $t$ and semi-adjacent parameter $a$. 

To compare the tradeoff functions, consider the composition first strategy. By the compsition property of DP, the tradeoff function of $(M_{1}(X),\cdots,M_{k}(X))$ for any databases $X,X'$ such that $A_{1}(X,X')= 1$, we have
\begin{equation*}
    T\left((M_{1}(X),\cdots,M_{k}(X)),(M_{1}(X'),\cdots,M_{k}(X'))\right) \geq f_{1}\otimes \cdots \otimes f_{k}.
\end{equation*}

Moreover, as $d(X,X') =a,$  by the group privacy property on $(M_{1}(\cdot),\cdots,M_{k}(\cdot))$, we have
\begin{equation*}
    T\left((M_{1}(X),\cdots,M_{k}(X)),(M_{1}(X'),\cdots,M_{k}(X'))\right) \geq (f_{1}\otimes \cdots \otimes f_{k})^{a}.
\end{equation*}

On the other hand, for the Semi-DP first strategy, note that for each $M_{i}(X)$, we have 
\begin{equation*}
    T(M_{i}(X),M_{i}(X')) \geq f^{\circ a},
\end{equation*}
for any pair $(X,X') \in IND(\mathcal{D}_{t},a)$ by Proposition \ref{thm: full characterize}. Then, by the composition property, the tradeoff function of the joint release $(M_{1}(X),\cdots,M_{k}(X))$ for $(X,X') \in IND(\mathcal{D}_{t},a)$ satisfies
\begin{equation*}
    T\left((M_{1}(X),\cdots,M_{k}(X)),(M_{1}(X'),\cdots,M_{k}(X'))\right) \geq f_{1}^{\circ a} \otimes \cdots \otimes f_{k}^{\circ a}.
\end{equation*}

Finally, the comparison between the two tradeoff functions is immediate by Lemma \ref{lemma: tradeoff function}.
\end{proof}

\subsection{Proof of Proposition \ref{thm: naive SS large}}
\begin{proof}
It is sufficient to show the set inclusions $S_{Semi} \subseteq S_{Group}$ and $\hull(S_{Semi}) \subseteq a\hull(S_{DP})$ since the inequalities in sensitivity immediately follows from the set inclusions.

To begin with, $S_{Semi} \subseteq S_{Group}$ is clear since $\mathcal{D}_{t} \subseteq \mathcal{D}$ and both the sensitivity space relies on the same adjacent function $A_{a}$.

To see $\hull(S_{Semi}) \subseteq a\hull(S_{DP})$, take any $u \in \hull(S_{Semi})$. Then there exist $s_{1},\cdots,s_{n}$ such that 

    \begin{equation*}
        u = \sum_{i=1}^{n}\lambda_{i}s_{i},
    \end{equation*}
    where the real numbers $\lambda_{i}$ satisfy $\lambda_{i} \geq 0$ and $\lambda_{1} + \cdots + \lambda_{n} = 1$.
    
    Note then for each $s_{i}$, $s_{i} = \phi(X)-\phi(X')$ for some $X,X'$ such that $d(X,X') \leq a$. Take $X_{0},X_{1},\cdots,X_{a}$ such that $X_{0} = X, X_{a} = X'$ and for each $i$ such that $0 \leq i \leq a-1$, $X_{i+1}$ be obtained from $X_{i}$ by changing one row. Then 
    \begin{equation*}
        \begin{split}
            s_{i} &= \phi(X_{0}) - \phi(X_{a}) \\
            &= \sum_{j=1}^{a}\left[\phi(X_{j-1}) - \phi(X_{j})\right]\\
            &=q_{1}^{i}+\cdots + q_{a}^{i},
        \end{split}
    \end{equation*}
    where $q_{j}^{i}:= \left[\phi(X_{j-1}) - \phi(X_{j})\right] \in S_{DP}$ for all $j \in \{1,2,\cdots,a\}$.
    
    By taking $ \beta_{ij} = \lambda_{i} $, we can write $u$ by
    \begin{equation*}
        \begin{split}
            u &= \sum_{i=1}^{n}\sum_{j=1}^{a}\lambda_{i}q_{j}^{i}
            = \sum_{i=1}^{n}\sum_{j=1}^{a}\beta_{ij}q_{j}^{i},
        \end{split}
    \end{equation*}
    where $\sum_{j=1}^{a}\sum_{i=1}^{n}\beta_{ij} = \sum_{j=1}^{a}\sum_{i=1}^{n}\beta_{ij}\lambda_{i} = \sum_{j=1}^{a} 1 =a$.
    Therefore, $\frac{u}{a} \in \hull(S_{DP})$ as desired
\end{proof}

\subsection{Proof of Proposition \ref{prop: linear invariant rank def}}
\begin{proof}
    Let \( t = L\phi(X) \), where \( L \) is a matrix representing a non-trivial linear transformation. The sensitivity space \( S_{Semi} \) is defined as:
    \[
    S_{Semi} = \{\phi(X) - \phi(X') : L\phi(X) = L\phi(X') \text{ and } A_{a(t)}(X, X') = 1\}.
    \]
    This means that for any \( u = \phi(X) - \phi(X') \in S_{Semi} \), the condition \( L\phi(X) = L\phi(X') \) implies \( L(\phi(X) - \phi(X')) = L(u) = 0 \). Hence, every element of \( S_{Semi} \) lies in the kernel of \( L \), so \( L(u) = 0 \) for all \( u \in S_{Semi} \).

    Since \( L \) is a non-trivial linear transformation, the kernel of \( L \), \( \text{ker}(L) \), has dimension strictly less than the full space \( \mathcal{D} \), implying that the space \( S \subseteq \text{ker}(L) \) is rank-deficient.
\end{proof}

\subsection{Proof for Proposition \ref{thm: projection framework}}

Before we delve into the main part of the proof of Proposition \ref{thm: projection framework}, the following Lemma \ref{thm: orthogonal complement indep of data} is the key tool for mechanism design in that it shows that applying $\Proj_{\mathcal{S}}^{\perp}\phi(\cdot)$ as post-processing does not affect the privacy guarantee, where $\Proj_{\mathcal{S}}^{\perp}\phi(\cdot)$ is the orthogonal projection operator onto the orthogonal complement of $\mathcal{S}$.

\begin{lemma}\label{thm: orthogonal complement indep of data}
    A mechanism $\text{Proj}_{\mathcal{S}}^{\perp}\phi(X)$ is constant for all input databases $X$.
\end{lemma}
\begin{proof}
For any $X,X'$ such that $A(X,X')=1$, we have 
\begin{equation*}
    \phi(X) - \phi(X') \in S,
\end{equation*}
so we have $\text{Proj}_{\mathcal{S}}^{\perp}(\phi(X)-\phi(X'))=0$, or $\text{Proj}_{\mathcal{S}}^{\perp}\phi(X) = \text{Proj}_{\mathcal{S}}^{\perp}\phi(X')$.

    Therefore, $\text{Proj}_{\mathcal{S}}^{\perp}\phi(X)$ is constant for all $X$
\end{proof}

\begin{proof}{\bf of Proposition \ref{thm: projection framework}} We apply post-processing property twice. 
    
First, apply a transformation $\Proj_{\mathcal{S}}$ to $M(X) = \phi(X) +e$. Then $\Proj_{\mathcal{S}}M(X) = \Proj_{\mathcal{S}}\phi(X) + \Proj_{\mathcal{S}}e $ still satisfies $(\mathcal{D},A,f)$-DP by the post-processing property. 

    Second, we apply an additional transformation on $\Proj_{\mathcal{S}}M(X)$ by adding  $\Proj_{\mathcal{S}}^{\perp}\phi(X)$. Recall that Lemma \ref{thm: orthogonal complement indep of data} shows $\Proj_{\mathcal{S}}^{\perp}\phi(X)$ is a data-independent constant. Therefore, by the post-processing property again,
    \begin{center}
    $\phi(X) + \Proj_{\mathcal{S}}e = \Proj_{\mathcal{S}}M(X) + \Proj_{\mathcal{S}}^{\perp}\phi(X)= \Proj_{\mathcal{S}}\phi(X) + \Proj_{\mathcal{S}}^{\perp}\phi(X)+\Proj_{\mathcal{S}}e $    
    \end{center}
     satisfies $(\mathcal{D},A,f)$-DP, as desired.
\end{proof}

\subsection{Proof for Proposition \ref{thm: gaussian mech guarantee}}

Note that we use the following lemma in this proof:
\begin{lemma}\label{thm: gaussian tradeoff function lemma}(Lemma A.2. in \citep{kim2022differentially}). Let $\mu_{1},\mu_{2} \in \mathbb{R}^{d}$ be arbitrary and $\Sigma$ be a $d \times d$ symmetric positive-definite matrix. Then,
\begin{equation*}
    T\left(N_{d}(\mu_{1},\Sigma),N_{d}(\mu_{2},\Sigma)\right) = G_{\Vert \Sigma^{-\frac{1}{2}}(\mu_{2}-\mu_{1}) \Vert_{2}}.
\end{equation*}
\end{lemma}

\begin{proof}
    We first show that the intermediary Gaussian mechanism 
\begin{center}
$M_{I}(X) = \phi(X) + e,$    
\end{center}
where $e \sim N_{d}\left(0,\left(\frac{\Delta_{2}(S_{Semi})}{\mu}\right)^{2}I_{d}\right)$ satisfies $(\mathcal{D}_{t},A_{a(t)},G_\mu)$-DP. 
     
To this end, take any pair of databases $X,X' \in \mathcal{D}_{t}$ such that $A_{a(t)}(X,X')=1$. Then $M_{I}(X) \sim N_{d}\left(\phi(X),\left(\frac{\Delta_{2}(S_{Semi})}{\mu}\right)^{2}I_{d}\right)$ and $M_{I}(X') \sim N_{d}\left(\phi(X'),\left(\frac{\Delta_{2}(S_{Semi})}{\mu}\right)^{2}I_{d}\right)$. By Lemma \ref{thm: gaussian tradeoff function lemma},
    \begin{equation*}
        T(M_{I}(X),M_{I}(X')) = G_{\Vert \phi(X) - \phi(X') \Vert_{2}\frac{\mu}{\Delta_{2}(S_{Semi})}}.
    \end{equation*}
    Since $0 \leq a \leq b$ if and only if $G_{a} \geq G_{b}$,
    \begin{equation*}
        \inf_{X,X': A(X,X')=1}T(M_{I}(X),M_{I}(X')) = \inf_{u \in S}G_{\Vert u \Vert_{2}\frac{\mu}{\Delta_{2}(S_{Semi})}} = G_{\sup_{u \in S}\Vert u \Vert_{2}\frac{\mu}{\Delta_{2}(S_{Semi})}} = G_{\mu}.
    \end{equation*}
    Therefore, $M_{I}(X)$ is $(\mathcal{D}_{t},A_{a(t)},G_\mu)$-DP.

    Now let $P$ be a projection matrix that projects onto $\Span(S)$. Then we can decompose $\phi(X)$ as $\phi(X) = P\phi(X) + (I-P)\phi(X)$. 
    
    Note that $PM_{I}(X) = P\phi(X) + Pe$ satisfies $(\mathcal{D}_{t},A_{a(t)},G_\mu)$-DP as well, by post-processing, where $Pe \sim N_{d}\left(0, (\Delta_{2}(S_{Semi}/\mu)^{2}P\right)$ so that $Pe \stackrel{d}{=} N$ in Algorithm \ref{alg: gaussian mechanism}. 

    Moreover, an additional transformation on $PM_{I}(X)$ by adding $(I-P)\phi(X)$ also satisfies $(\mathcal{D}_{t},A_{a(t)},G_\mu)$-DP by Lemma \ref{thm: orthogonal complement indep of data} that shows that $(I-P)\phi(X)$ is a data-independent constant. Altogether, $PM_{I}(X) + (I-P)\phi(X) = P\phi(X) + (I-P)\phi(X) + N = \phi(X) + N$ satisfies $(\mathcal{D}_{t},A_{a(t)},G_\mu)$-DP, as desired.
 
\end{proof}

\subsection{Proof for Lemma \ref{thm: convex hull norm ball}}

\begin{proof}
    Note that a set $K$ is a norm ball in $\mathcal{S}$ if $K \in \mathcal{S}$ is 1) convex, 2) bounded 3) symmetric about zero: if $u \in K,$ then $-u \in K$, and 4) absorbing: $\forall u \in \mathcal{S}, \exists c >0$ such that $u \in cK$.

    Note that $K = \hull(S)$ is convex by the definition of convex hull and the boundedness condition is straightforward as $S$ is assumed to be bounded. In addition, symmetry is clear from the symmetry of adjacency relation in DP.

To see the absorbing property, take any $v \in \mathcal{S}$. Since $S$ spans $\mathcal{S}$, we can write $v = \sum_{i=1}^{m}a_{i}s_{i}$, where $m = \vert \mathcal{S} \vert < \infty$.

Now, let $n$ be the number of nonnegative $a_{i}'s$ and rearrange the summation by the sum of nonnegative coefficients and negative coefficients:
\begin{equation*}
\begin{split}
    v &= \sum_{i=1}^{n}a_{i}s_{i} + \sum_{j=n+1}^{m}a_{j}s_{j} \\
    &= \sum_{i=1}^{n}a_{i}s_{i} + \sum_{j=n+1}^{m}\vert a_{j}\vert (-s_{j}) \\
    &= \left(\sum_{i=1}^{n}a_{i}\right)\sum_{i=1}^{n}\frac{a_{i}}{\sum_{i=1}^{n}a_{i}}s_{i} + \left(\sum_{j=n+1}^{m}\vert a_{j}\right) \sum_{j=n+1}^{m}\frac{\vert a_{j} \vert}{\sum_{j=n+1}^{m}\vert a_{j} \vert}(-s_{j}).
\end{split}
\end{equation*}
We see that $h_{1} := \sum_{i=1}^{n}\frac{a_{i}}{\sum_{i=1}^{n}a_{i}}s_{i}$ and $h_{2} := \sum_{j=n+1}{m}\frac{\vert a_{j} \vert}{\sum_{j=n+1}^{m}\vert a_{j} \vert}(-s_{j}) $ are both convex combinations of $S$. Now letting $k = \frac{1}{\sum_{i=1}^{n}a_{i}+\sum_{j=n+1}^{m}\vert a_{j} \vert}$, then 
\begin{equation*}
    kv = \frac{\sum_{i=1}^{n}a_{i}}{\sum_{i=1}^{n}a_{i}+\sum_{j=n+1}^{m}\vert a_{j} \vert}h_{1} + \frac{\sum_{j=n+1}^{m}\vert a_{j} \vert}{\sum_{i=1}^{n}a_{i}+\sum_{j=n+1}^{m}\vert a_{j} \vert}h_{2},
\end{equation*}
so $kv$ is a convex combination of $h_{1}$ and $h_{2}$. Thus, $kv \in \hull(S)$.

Therefore, $v \in c\hull(S)$, where $c = 1/k$ and we have that the norm ball is absorbing.
\end{proof}

\subsection{Proof for Theorem \ref{thm: semi dp K norm mech}}

\begin{proof}
We first prove the privacy guarantee and then prove the optimality. To begin with, consider the following decomposition of $\phi(X)$:

\begin{equation*}
    \phi(X) = \Proj_{\mathcal{S}}\phi(X) + \Proj_{\mathcal{S}}^{\perp}\phi(X).
\end{equation*}

Since $\mathcal{S}$ is a subspace of $\mathbb{R}^{d}$ of dimension $s$, it is isomorphic to $\mathbb{R}^{s}$. For an isomorphism $\theta: \mathcal{S} \rightarrow \mathbb{R}^{s}$, consider $\theta(\Proj_{\mathcal{S}}\phi(X)) \in \mathbb{R}^{s}$. It is clear that a mechanism $\theta(\Proj_{\mathcal{S}}\phi(X)) +V$ satisfies $(\mathcal{D},A,f_{\epsilon})$-DP. Then the privacy guarantee for $\phi(X) + \theta^{-1}(V)$ is by applying the post-processing twice: First, the post-processing $\Proj_{\mathcal{S}}\phi(X) + \theta^{-1}(V) = \theta^{-1}(\theta(\Proj_{\mathcal{S}}\phi(X)) +V)$ satisfies $(\mathcal{D},A,f_{\epsilon})$-DP as well. Moreover, $\Proj_{\mathcal{S}}^{\perp}\phi(X)$ is a data-independent constant for all $X$ by Lemma \ref{thm: orthogonal complement indep of data} so post-processing ensures again that $\phi(X) + \theta^{-1}(V)= \Proj_{\mathcal{S}}\phi(X) + \Proj_{\mathcal{S}^{\perp}}\phi(X) +\theta^{-1}(V)$ is $(\mathcal{D},A,f_{\epsilon})$-DP.

To demonstrate the optimality of our \(K\)-norm mechanism, we begin by establishing the minimality of the norm ball in the subspace \(\mathcal{S} = \text{span}(S)\) and then extend this result to the space \(\mathbb{R}^s\) using the isomorphism \(\theta: \mathcal{S} \rightarrow \mathbb{R}^s\). In \(\mathcal{S}\), the sensitivity space \(S\) is bounded and rank-deficient, with rank \(s < d\). The convex hull \(K = \hull(S)\) forms the minimal convex set containing \(S\). Since \(K\) is a valid norm ball in \(\mathcal{S}\) by Lemma \ref{thm: convex hull norm ball}, it is the smallest convex set that captures the variations of the sensitivity space. Consequently, any other norm ball in \(\mathcal{S}\) that contains \(S\) must also contain \(K\), making \(\Delta_{K} \cdot K\) the minimal norm ball with respect to the containment order among all norm balls in \(\mathcal{S}\). This minimality ensures that the K-norm mechanism defined by \(\Vert \cdot \Vert_{K}\) is optimal in \(\mathcal{S}\), as it minimizes the noise required to maintain the privacy guarantees by aligning the noise precisely with the sensitivity space.

The isomorphism \(\theta\) between \(\mathcal{S}\) and \(\mathbb{R}^s\) is a linear bijection that preserves convexity, containment, and other geometric properties. By applying \(\theta\) to the convex hull \(K\), we obtain the set \(K_s = \theta(K)\) in \(\mathbb{R}^s\). Since \(\theta\) preserves the containment order, the minimality of \(K\) in \(\mathcal{S}\) directly implies the minimality of \(K_s\) in \(\mathbb{R}^s\). Therefore, the norm ball \(\Delta_{s} \cdot K_s\) is the smallest convex set containing \(\theta(S)\) in \(\mathbb{R}^s\). For any other norm \(\Vert \cdot \Vert_{H}\) in \(\mathbb{R}^s\) with norm ball \(\Delta_H \cdot H\), we have \(\Delta_{s} \cdot K_{s} \subseteq \Delta_H \cdot H\), indicating that \(\Delta_{s} \cdot K_{s}\) optimally contains \(\theta(S)\).
\end{proof}

\subsection{Proof of Theorem \ref{thm: table sens space}}

\begin{proof}
Without loss of generality, consider an individual \( x \) in the dataset with features \( A_1 = i \) and \( A_2 = j \), who attempts to pretend to be another individual \( y \), where \( A_1 = l \) and \( A_2 = k \). For the adversary to not detect the change, the feature combination \( (l, k) \) must correspond to a non-zero count in the contingency table. Specifically, this means \( \#\{i: X_1^{(i)} = l\} \neq 0 \) and \( \#\{i: X_2^{(i)} = k\} \neq 0 \). If these conditions are not satisfied, the adversary can immediately detect that \( x \) is lying, as such an individual \( y \) cannot exist based on the true counts provided by the invariant.

In the case where \( x_{lk} > 0 \), replacing \( x \) with \( y \) and simultaneously switching one of the $x_{lk}$ individuals in the dataset having features \( A_1 = l \) and \( A_2 = k \) with another individual having the feature \( A_1 = i \) and \( A_2 = j \) results in no net change in the contingency table. Specifically, such a switch preserves both the row and column sums, ensuring that the one-way margins remain unchanged. In this case, the number of changes is 2 and the corresponding element for the sensitivity space is \( \mathbf{0} \).

The worst-case scenario, which aligns with the upper bound on \( a(t) = 3 \), occurs when \( x_{lk} = 0 \). In this situation, to maintain the one-way margins, one individual from the category \( A_1 = l \) and \( A_2 = j \) must be replaced with an individual with \( A_1 = i \) and \( A_2 = j \). Similarly, one individual from the category \( A_1 = i \) and \( A_2 = k \) must be replaced with another individual having \( A_1 = i \) and \( A_2 = j \). It is important to note that any other choice of switching results in more than 3 changes, so this switching process is the minimal one required to maintain the one-way margins. The resulting changes in the contingency table are as follows: \( x_{ij} \) increases by 1, \( x_{ik} \) decreases by 1, \( x_{lj} \) decreases by 1, and \( x_{lk} \) increases by 1.

Alternatively, if we consider an individual \( x \) with features \( A_1 = i \) and \( A_2 = k \), attempting to pretend to be \( y \) with features \( A_1 = l \) and \( A_2 = j \), a similar process occurs. By switching individuals as outlined above, the resulting changes in the contingency table are \( x_{ij} - 1 \), \( x_{ik} + 1 \), \( x_{lj} + 1 \), and \( x_{lk} - 1 \).

Consequently, we conclude that the sensitivity space is 
\[ S_{Semi} = \left\{\mathbf{v}_{ijlk} \in \mathbb{R}^{r \times c} : i, l \in \{1, \ldots, r\}, \; j, k \in \{1, \ldots, c\}, \; i \neq l, \; j \neq k \right\} \cup \{\mathbf{0}\}. \]
\end{proof}

\subsection{Proof of Theorem \ref{thm: semi dp umpu}}
To prove Theorem \ref{thm: semi dp umpu}, we recall the neyman structure and its connection to unbiased tests.

\begin{definition}\label{def: neyman structure} (Definition 4.120 of \cite{schervish2012theory}). For a parameter space $\Omega$, let $G \subset \Omega$. If $T$ is a sufficient statistic for $G$, then a test $\phi$ has Neyman structure relative to $G$ and $T$ if $\mathbb{E}_{w}[\phi(X) \vert T=t]$ is constant in $t$ for all $w \in G$. 
\end{definition}

\begin{lemma}\label{thm: neyman umpu schervish} (Theorem 4.123 of \cite{schervish2012theory}). Let $\Omega = \Omega_{0} \cup \Omega_{1}$ be a partition. Let $G = \Bar{\Omega}_{0} \cap \Bar{\Omega}_{1}$, where $\Bar{\Omega}_{0}$ is the closure of $\Omega$. Let $T$ be a boundedly complete sufficient statistic for $G$. Assume that the power function is continuous. If there is a UMPU level $\alpha$ test $\phi$ among those which have Neyman structure relative to $G$ and $T$, then $\phi$ is UMPU level $\alpha$.
\end{lemma}

\begin{lemma}\label{thm: neyman umpu size alpha} Consider the $2 \times 2$ table and the hypothesis $H_{0}$: $w \leq 1$ vs. $H_{1}:$ $w >1$. Let $\Phi$ be a set of tests. If there exists a UMP $\phi \in \Phi$ among those with 
\begin{equation*}
\mathbb{E}_{H \sim \text{Hyper}(t_{2 \cdot},t_{1 \cdot},t_{\cdot 1})}\phi(t_{\cdot 1}-H, t_{1\cdot}-t_{\cdot 1}+H,H,t_{2\cdot}-H) = \alpha,
\end{equation*}
for all $\alpha$, then $\phi$ is UMPU size $\alpha$ among $\Phi$.
\end{lemma}

The following lemma is an adaptation of the Neyman-Pearson lemma for \( f \)-DP testing, as introduced in \cite{awan2023canonical}, applied to our context of Semi-DP by clarifying dataspace and adjacency function.

\begin{lemma}\label{thm: dp np lemma} (Theorem 4.8 in \cite{awan2023canonical}).
Let \( f \) be a symmetric nontrivial tradeoff function, and let \( F \) be a CND for \( f \). Consider \( \mathcal{X} = \{0,1\} \) and \( \mathcal{D} = \mathcal{X}^{n} \). Let \( P \) and \( Q \) be two exchangeable distributions on \( \mathcal{X}^{n} \) with pmfs \( p \) and \( q \), such that \( \frac{q}{p} \) is an increasing function of \( x = \sum_{i=1}^{n} x_{i} \). Let \( \alpha \in (0,1) \). Given that \( |x - x'| \leq 1 \), where \( x \) and \( x' \) are from datasets \( X \) and \( X' \in \mathcal{D} \), under the adjacency function \( A \), the most powerful \( (\mathcal{X}^{n}, A, f) \)-DP test \( \phi \) with significance level \( \alpha \) for testing \( H_{0}: X \sim P \) vs. \( H_{1}: X \sim Q \) can be expressed in any of the following forms:
\begin{itemize}
    \item[1. ] There exists $y \in \{0,1,2,\cdots,n\}$ and $c \in (0,1)$ such that for all $x \in \{0,1,2,\cdots,n\},$
    \begin{equation*}
        \phi(x) = (1-f(1-\phi(x-1)))\mathbbm{1}(x>y) + c \mathbbm{1}(x=y),
    \end{equation*}
    where if $y>0$ then $c$ satisfies $c \leq 1-f(1)$, and $c$ and $y$ are chosen such that $\mathbb{E}_{P}\phi(x) = \alpha$. If $f(1)=1$, then $y=0$.
    \item[2. ] $\phi(x) = F(x-m),$ where $m \in \mathbb{R}$ is chosen such that $\mathbb{E}_{P}\phi(x)=\alpha$.
    \item[3. ] Let $N \sim F$. The variable $T = X+N$ satisfies $f$-DP. Then $p = \mathbb{E}_{X \sim P}F(X-T)$ is a $p$-value and $\mathbbm{1}(p \leq \alpha) \vert X = \mathbbm{1}(T \geq m) \vert X \sim \text{Bern}(\phi(X))$, where $\phi(x)$ agrees with $1$ and $2$ above.
\end{itemize}
\end{lemma}

\begin{proof}{\bf of Theorem \ref{thm: semi dp umpu}}.
We rely on Lemma \ref{thm: neyman umpu schervish} to establish the existence of a uniformly most powerful unbiased (UMPU) test among $\Phi_{f}^{\text{semi}}$. By Lemma \ref{thm: neyman umpu schervish}, we know that if a test has Neyman structure and is UMPU at level \( \alpha \) among a set of unbiased tests, then it is the UMPU level \( \alpha \) test. In our case, the parameter of interest \( w \) is the odds ratio, and we partition the parameter space as \( \Omega = \Omega_0 \cup \Omega_1 \), where \( \Omega_0 \) represents the null hypothesis region \( w \leq 1 \) and \( \Omega_1 \) represents the alternative hypothesis \( w > 1 \). Therefore, by Lemma \ref{thm: neyman umpu schervish}, it is sufficient to consider tests which have Neyman structure relative to $w=1$ and $t = (t_{1\cdot},t_{2\cdot},t_{\cdot 1},t_{\cdot 2}) = (x_{11}+x_{12},x_{21}+x_{22},x_{11}+x_{21},x_{12}+x_{22})$. That is, we only consider the tests among $\Phi_{f}^{\text{semi}}$. 

To apply Lemma \ref{thm: dp np lemma}, note that $x_{11}$ is the count and $x_{21}$ can  be differ by $1$ for adjacent databases $X,X'$ satisfying that $A_{a(t)}(X,X')=1$, guaranteed by Equation \ref{equ: f-semi dp condition}. For the tradeoff function $f$, by Lemma \ref{thm: dp np lemma}, there exists a most powerful $\Psi_{f}$ test for $H_{0}: w=1$ vs. $H_{1}: w=w_{1}(>1)$, which is of the form 
\[
\psi_{z}^{\ast}(x)=F_{f}(x-m(t)),
\]
$\psi_{z}^{\ast}(x)=F_{f}(x-m(t))$, where $m(t)$ is chosen such that $\mathbb{E}_{H \sim \text{Hyper}(t_{2\cdot},t_{1\cdot},t_{\cdot 1},1)}F_{f}(H-m) = \alpha$. Since this test does not depend on the specific alternative, it is UMP for $H_{0}: w \leq 1$ vs. $H_{1}: w >1$. The $p$-value also followes from Lemma \ref{thm: dp np lemma}.
\end{proof}

\section{Details on Simulation}\label{app: details on simulation}
In this section, we provide implement details on $\ell_{1},\ell_{2},\ell_{\infty}$ and $K$-norm mechanisms in Section \ref{sec: contingency numerical}. 

\begin{algorithm}[htb!]
\caption{\texttt{Sampling from $\ell_{1}$-mechanism} }\label{alg: l1 mech}
\begin{algorithmic}[1]
    \STATE \textbf{Input:} Query $\phi$ to be privatized, $\Delta_{1}(S),$ and $\epsilon$
    \STATE Set $d = \text{dim}(\phi(X))$
    \STATE Draw $V_{j} \stackrel{iid}{\sim} \text{Laplace}\left(\frac{\Delta_{1}(S)}{\epsilon}\right)$, for $j=1,\cdots,d$
    \STATE Set $V = (V_{1},\cdots, V_{d})^\top$
    \STATE \textbf{Output: } $\phi(X) + V$
\end{algorithmic}    
\end{algorithm}

\begin{algorithm}[htb!]
\caption{\texttt{Sampling from $\ell_{2}$-mechanism} \citep{wang2014entropy}}\label{alg: l2 mech}
\begin{algorithmic}[1]
    \STATE \textbf{Input:} Query $\phi$ to be privatized, $\Delta_{2}(S),$ and $\epsilon$
    \STATE Set $d = \text{dim}(\phi(X))$
    \STATE Draw $Z \sim N(0, I_{d})$
    \STATE Draw $r \sim \text{Gamma}(\alpha = d, \beta = \epsilon/\Delta_{2}(S))$
    \STATE Set $V = \frac{rZ}{\Vert Z \Vert_{2}}$
    \STATE \textbf{Output: } $\phi(X) + V$
\end{algorithmic}    
\end{algorithm}

\begin{algorithm}[htb!]
\caption{\texttt{Sampling from $\ell_{\infty}$-mechanism }\citep{steinke2016between}}\label{alg: linfty mech}
\begin{algorithmic}[1]
    \STATE \textbf{Input:} Query $\phi$ to be privatized, $\Delta_{\infty}(S),$ and $\epsilon$
    \STATE Set $d = \text{dim}(\phi(X))$
    \STATE Set $U_{j} \stackrel{iid}{\sim} U(-1,1)$ for $j=1,\cdots,d$
    \STATE Draw $r \sim \text{Gamma}(\alpha = d+1, \beta = \epsilon/\Delta_{\infty}(S))$
    \STATE Set $V = r \cdot (U_{1},\cdots, U_{d})^\top$
    \STATE \textbf{Output: } $\phi(X) + V$
\end{algorithmic}    
\end{algorithm}

\vskip 0.2in
\bibliography{JMLRtemplate/ref}
\end{document}